\documentclass[journal,onecolumn]{IEEEtran}
\usepackage{bbm}
\usepackage{cite}
\usepackage{amsfonts}
\usepackage{mathrsfs}
\usepackage{graphicx}
\usepackage{amssymb}
\usepackage{latexsym}
\usepackage{amsmath}
\usepackage{stfloats}
\usepackage{cases}
\usepackage{setspace}
\usepackage{bigstrut}
\usepackage{bm}
\usepackage{url}
\usepackage{color}

\renewcommand{\baselinestretch}{1.46}

\newtheorem{theorem}{Theorem}

\newtheorem{corollary}{Corollary}

\begin{document}

\title{Efficient Downlink Channel Reconstruction for FDD Multi-Antenna Systems}
\author{Yu Han$^\ast$, Tien-Hao Hsu$^\dagger$, Chao-Kai Wen$^\dagger$, Kai-Kit~Wong$^\S$, and Shi Jin$^\ast$\\
$^\ast$National Mobile Communications Research Laboratory, Southeast University, Nanjing, China\\
$^\dagger$Institute of Communications Engineering, National Sun Yat-sen University, Kaohsiung 804, Taiwan\\
$^\S$Department of Electronic and Electrical Engineering, University College London, United Kingdom}
\maketitle

\vspace{-5mm}

\begin{abstract}
In this paper, we propose an efficient downlink channel reconstruction scheme for a frequency-division-duplex multi-antenna system by utilizing uplink channel state information combined with limited feedback. Based on the spatial reciprocity in a wireless channel, the downlink channel is reconstructed by using frequency-independent parameters. We first estimate the gains, delays, and angles during uplink sounding. The gains are then refined through downlink training and sent back to the base station (BS). With limited overhead, the refinement can substantially improve the accuracy of the downlink channel reconstruction. The BS can then reconstruct the downlink channel with the uplink-estimated delays and angles and the downlink-refined gains. We also introduce and extend the Newtonized orthogonal matching pursuit (NOMP) algorithm to detect the delays and gains in a multi-antenna multi-subcarrier condition. The results of our analysis show that the extended NOMP algorithm achieves high estimation accuracy. Simulations and over-the-air tests are performed to assess the performance of the efficient downlink channel reconstruction scheme. The results show that the reconstructed channel is close to the practical channel and that the accuracy is enhanced when the number of BS antennas increases, thereby highlighting that the promising application of the proposed scheme in large-scale antenna array systems.
\end{abstract}


\section{Introduction}\label{Sec:Introduction}

Frequency division duplex (FDD) is one of the most widely used duplexing modes for mobile communication systems where both directions of communication can take place simultaneously without interference. The FDD mode in multiple-input multiple-output (MIMO) antenna systems has achieved great success in 3G and 4G commercial mobile communication networks. Recently, large-scale or massive MIMO, which is capable of using spatial dimensions to guarantee an extraordinary spectral efficiency, has been identified as a key enabler for 5G networks.

However, the use of large-scale antenna arrays in 5G \cite{3GPPTS38201,3GPPTS38912,3GPPTS38913,Andrews2014,Larsson2014} and future networks imposes huge challenges to the acquisition of downlink channel state information (CSI) at the base station (BS) for FDD-MIMO systems, and such information is crucial to an excellent system performance especially in the downlink. The lack of reciprocity between the uplink and downlink channels on different frequency bands makes the downlink CSI acquisition process difficult to achieve. Therefore, downlink CSI is usually acquired through downlink training and feedback.

In previous generations of networks, only a few antennas are used at the BS. An abundant amount of time and frequency resources are available to form orthogonal pilots and the amount of feedback is relatively small. While in 5G and future networks, the use of hundreds or even thousands of antenna ports prevents the design of completely orthogonal pilot patterns. In this case, reusing pilots becomes inevitable \cite{Larsson2014,Chen2017,Marzetta2010}, thereby affecting the accuracy of the CSI estimation. Moreover, using a huge amount of feedback for a high-dimensional complex channel matrix is impractical. Thus, downlink CSI acquisition remains a key problem in FDD massive MIMO systems.

Many studies have been conducted to address the aforementioned problem \cite{Choi2015,Shen2017,Choi2014,Gao2015,Rao2014,Hugl2002,Xie2017,Han2017,Khalilsarai2017,Shen2016,Vasisht2016Decimeter,Vasisht2016Eliminating}. These studies generally apply two types of approaches. In the first approach, downlink CSI is solely obtained from downlink training and feedback but does not require orthogonality among the pilots that are transmitted from different antennas. For example, codebooks are used to quantize the space, and only the codebook indices are sent back to the BS. In \cite{Choi2015} and \cite{Shen2017}, trellis-based and angle-of-departure-adaptive subspace codebooks were proposed, respectively, to quantize the channel of FDD massive MIMO systems. These methods require training and produce feedback overhead. Other methods have attempted to exploit the slow-varying nature of the space. In \cite{Choi2014}, the authors assumed that the channels were correlated in both time and space; they also proposed an open- and closed-loop training with CSI memory that could be derived from previous time instances. If the channel is sparse, then compressed sensing can effectively reduce the training and feedback overhead \cite{Gao2015,Rao2014}.

In the second approach, the spatial reciprocity between channels on two close frequency bands is applied. In \cite{Hugl2002}, the authors validated the spatial congruence by conducting measurements and demonstrated a small deviation in the dominant directions of arrival at the frequencies 1935 MHz and 2125 MHz. Based on these findings, \cite{Xie2017} reconstructed the downlink channel with the aid of the downlink channel covariance matrix that was inferred from the uplink channel covariance matrix. Using the uplink CSI also helps reduce the downlink training and feedback overhead. For instance, \cite{Han2017} and \cite{Khalilsarai2017} proposed to exploit the channel sparsity for estimating the propagation directions via uplink training and used the direction estimates in the downlink training process to reduce the feedback overhead. In \cite{Shen2016}, the authors proposed a compressed downlink CSI acquisition method that uses the partial support information obtained from the uplink and demonstrated that this method could significantly reduce the training overhead. Nevertheless, the aforementioned downlink CSI acquisition methods have not been examined via over-the-air (OTA) tests. To verify the effectiveness of CSI on other bands in practical systems, \cite{Vasisht2016Decimeter} and \cite{Vasisht2016Eliminating} conducted OTA tests and found that the aforementioned method have promising application in inferring the RF channels on one band by using the CSI on another band. In \cite{Vasisht2016Eliminating}, the authors proposed to completely eliminate the downlink training and feedback in long-term evolution (LTE) systems because the gain of each path in a wireless channel was thought to be frequency-independent similar to delay and angle. However, no sufficient evidence can confirm the frequency-independent feature of the gains, which may greatly degrade the performance of the scheme proposed in \cite{Vasisht2016Eliminating} when the uplink and downlink frequency bands are distinctly separated.

Motivated by spatial reciprocity, this paper proposes an efficient downlink channel reconstruction scheme that utilizes the frequency-independent parameters of the delays and angles of the multipath channel for a FDD multi-antenna orthogonal frequency-division multiplexing (OFDM) system.\footnote{The proposed scheme can be straightforwardly applied to FDD massive MIMO systems.} Given the uncertainty of their frequency-independence, we further refine the gains by using a limited amount of downlink training and feedback. Our major contributions are summarized as follows.
\begin{itemize}
\item \emph{Downlink channel reconstruction}: An efficient downlink channel reconstruction scheme is developed. The frequency-independent delays and angles are initially estimated during the uplink training process by using the Newtonized orthogonal matching pursuit (NOMP) algorithm that is extended in this paper. Afterward, the gains are refined by using the least squares (LS) algorithm in the downlink with a small amount of pilots and feedback. With the uplink-estimated delays and angles as well as the downlink-refined gains, the downlink channel can be reconstructed at the BS. The necessity of downlink refinement is proven through theoretical analyses and simulations. The numerical results demonstrate that the proposed efficient downlink reconstruction scheme can be used to reconstruct a highly accurate downlink channel.
\item \emph{Extension of the NOMP algorithm}: NOMP was originally designed to extract two parameters, namely, gains and frequencies, for a noisy mixture of sinusoids \cite{Mamandipoor2016}. In this paper, we extend and adapt this algorithm to a trivariate case, where the gain, delay, and angle of each path are estimated. During each iteration of this algorithm, a 2D dictionary is utilized and the Newton step refines the delay and angle simultaneously. After updating the stopping criteria, we evaluate the accuracy of the estimations by deriving the lower bounds of the estimation errors and observe that the extracted delays and angles are very close to the real values.
\item \emph{OTA test results}: An OTA testbed is set up to assess the system performance of the proposed downlink reconstruction scheme in practical wireless communication scenarios. We observe that the channel reconstructed by the proposed scheme is near the linear minimum mean square error (LMMSE)-estimated channel, thereby demonstrating the necessity of gain refinement and the effectiveness of the reconstruction. The OTA results also show that with more antennas, the efficient channel reconstruction scheme demonstrates higher accuracy and can behave well in a massive MIMO scenario.
\end{itemize}

The rest of this paper is organized as follows. Section \ref{Sec:ChannelModel} introduces the multipath channel between the BS and a user and studies the frequency-independent spatial parameters over different frequency channels. Section \ref{Sec:Reconstruction} proposes an efficient downlink channel reconstruction scheme based on the uplink-estimated frequency-independent parameters and analytically justifies the importance of refining the gains in the downlink. Section \ref{Sec:NOMP} presents the extended NOMP algorithm for estimating the gains, delays, and angles through the uplink training process and analyzes its estimation accuracy. Section \ref{Sec:Results} discusses the simulation and OTA test results for the proposed efficient downlink channel reconstruction scheme. Section \ref{Sec:Conclusion} concludes the paper.

\emph{Notations}---In this paper, the matrices and vectors are denoted by uppercase and lowercase boldface letters, respectively, while the superscripts $(\cdot)^\dag$, $(\cdot)^{H}$, and $(\cdot)^{T}$ denote the pseudo-inverse, conjugate-transpose, and transpose, respectively. In addition, $\mathcal{R}\{\cdot\}$ takes the real component of a complex number, while $\mathbb{E}\{\cdot\}$ represents the expectation with respect to all random variables inside the brackets. We also use $\left| \cdot \right|$ and $\left\| \cdot \right\|$ to denote taking the absolute value and modulus operations, respectively, and use the notations $\left\lfloor \cdot \right\rfloor$ and $\left\lceil \cdot \right\rceil$ to denote rounding a decimal number to its nearest lower and higher integers, respectively.

\section{Channel Model}\label{Sec:ChannelModel}

In this section, we describe the wireless channel between the BS and its serving user by tracing the propagation paths of the signal. A single cell of a mobile communication system operates in the FDD mode by employing OFDM. We denote the difference between the uplink and downlink carrier frequencies by $\triangle F$ and assume that each uplink and downlink frequency band has $N$ sub-carriers with spacing $\triangle f$. We focus on the baseband and denote the uplink central sub-carrier by DC. The BS is equipped with a uniform linear array (ULA) with $M$ antenna elements, while the user has one antenna.

The uplink multipath channel between the user and the BS antenna element $m$ on subcarrier $n$ can be modeled as
\begin{equation}\label{Eq:SISOchannelmn1}
h^{\rm ul}_m (n)=\sum\limits_{l = 0}^{{L^{\rm ul}-1}} {g^{\rm ul}_l e^{j2\pi \Phi^{\rm ul}_{l,n} + j2\pi \Theta^{\rm ul}_{l,m}}},
\end{equation}
in which $m = -\lfloor{M}/{2}\rfloor,\dots,-1,0,1,\dots,\lceil{M}/{2}\rceil-1$, $n = -\lfloor{N}/{2}\rfloor,\dots,-1,0,1,\dots,\lceil{N}/{2}\rceil-1$, $L^{\rm ul}$ is the number of propagation paths in the uplink, $g^{\rm ul}_l$ is the gain of the $l$th propagation path in the uplink and is complex, $\Phi^{\rm ul}_{l,n}$ is introduced by the delay of the $l$th path in the uplink, and $\Theta^{\rm ul}_{l,m}$ is the phase difference between antenna element $m$ and $0$ resulting from the time difference in the arrival of the $l$th path in the uplink.

$\tau^{\rm ul}_l$ denotes the delay of the $l$th propagation path in the uplink upon its arrival at antenna element $0$, which satisfies $0<\tau^{\rm ul}_l<1/{\triangle f}$. We know that $\Phi^{\rm ul}_{l,n}=n \triangle f\tau^{\rm ul}_l$ and we further denote the angle of the $l$th path in uplink by $\theta^{\rm ul}_l$, which satisfies $0<\theta^{\rm ul}_l<2\pi$. The wireless signal travels different distances when arriving at different BS antenna elements as illustrated in Fig. \ref{Fig:ULA}. The signal from direction $\theta^{\rm ul}_l$ travels at a longer distance of $d^{\rm ul}_{l,m} = md\sin\theta^{\rm ul}_l$ upon arriving at element $m$ when compared with element $0$, where $d$ denotes the distance between two adjacent antenna elements that equals to ${\lambda}/{2}$ and $\lambda$ denotes the carrier wavelength. The phase difference between element $m$ and $0$ for the $l$th path is calculated as
\begin{equation}\label{Eq:Theta}
\Theta^{\rm ul}_{l,m} = \frac{f d^{\rm ul}_{l,m}}{c}=m\frac{d}{\lambda}\sin\theta^{\rm ul}_l,
\end{equation}
where $f$ denotes the carrier frequency and $c$ denotes the speed of light. Therefore, \eqref{Eq:SISOchannelmn1} can be rewritten as
\begin{equation}\label{Eq:SISOchannelmn}
h^{\rm ul}_m (n)=\sum\limits_{l = 0}^{{L^{\rm ul}-1}} {g_l e^{j2\pi n \triangle f\tau^{\rm ul}_l + j2\pi m\frac{d}{\lambda}\sin\theta^{\rm ul}_l}}.
\end{equation}

By stacking the channels on all subcarriers and antennas into a vector, we obtain the multi-subcarrier multi-antenna channel between the user and the BS, which is expressed as
\begin{equation}\label{Eq:SIMO-OFDMchannel}
{\bf h}^{\rm ul}=\sum\limits_{l = 0}^{{L^{\rm ul}-1}} {g^{\rm ul}_l {\bf p}(\tau^{\rm ul}_l) \otimes {\bf a}(\theta^{\rm ul}_l)},
\end{equation}
where $\otimes$ represents the Kronecker product,
\begin{equation}\label{Eq:pvec}
{\bf p}(\tau) = \left[e^{-j2\pi \lfloor\frac{N}{2}\rfloor \triangle f \tau},\dots, e^{j2\pi (\lceil\frac{N}{2}\rceil-1) \triangle f \tau}\right]^H
\end{equation}
denotes the delay-related phase vector of the OFDM module, and
\begin{equation}\label{Eq:avec}
{\bf a}(\theta) = \left[e^{-j2\pi \lfloor\frac{M}{2}\rfloor \frac{d}{\lambda}\sin\theta},\dots, e^{j2\pi (\lceil\frac{M}{2}\rceil-1) \frac{d}{\lambda}\sin\theta}\right]^H
\end{equation}
denotes the steering vector of the ULA.

As a special case, the BS only has one antenna. Under this single-input single-output (SISO) condition, the angles of propagation paths are not modeled in this channel, but the uplink channel vector on all subcarriers can be written as
\begin{equation}\label{Eq:SISOchannel}
{\bf h}^{\rm ul}_{\rm SISO}=\sum\limits_{l = 0}^{{L^{\rm ul}-1}} {g^{\rm ul}_l {\bf p}(\tau^{\rm ul}_l)}.
\end{equation}

\begin{figure}
  \centering
  \includegraphics[scale=0.50]{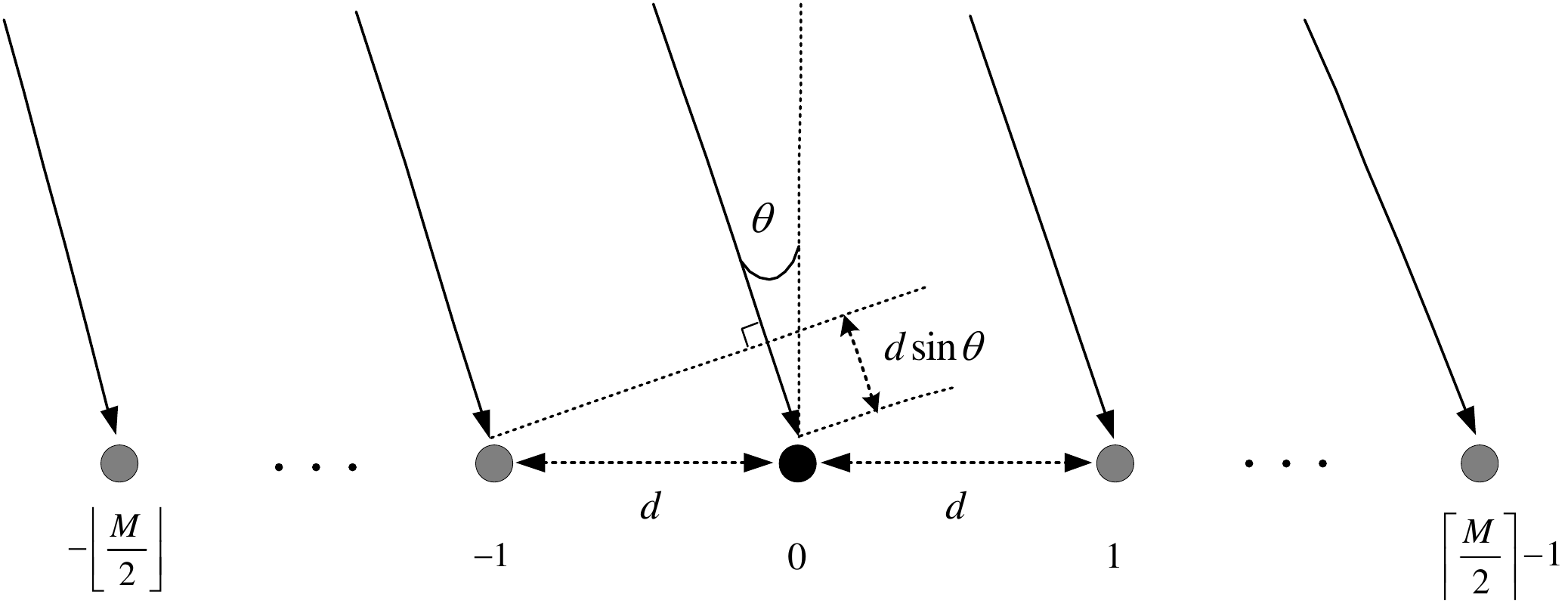}
  \caption{Difference of the propagation distances when the wireless signal arrives at two adjacent ULA elements. The black circle represents the reference antenna element at the BS.}\label{Fig:ULA}
\end{figure}

For the downlink, by using the uplink carrier frequency as the reference (i.e., $0$ Hz), we denote the downlink carrier frequency by $\triangle F$ and model the downlink channel between the BS antenna array and the user on all subcarriers and antennas as
\begin{equation}\label{Eq:MISOchanneln}
{\bf h}^{\rm dl}=\sum\limits_{l = 0}^{{L^{\rm dl}-1}} {g^{\rm dl}_l e^{j2\pi \triangle F \tau^{\rm dl}_l} {\bf p}(\tau^{\rm dl}_l) \otimes {\bf a}(\theta^{\rm dl}_l)},
\end{equation}
where $L^{\rm dl}$ represents the number of propagation paths in the downlink, $g^{\rm dl}_l$ is the complex gain of the $l$th downlink propagation path, $\tau^{\rm dl}_l$ is the delay of the $l$th downlink path with respect to antenna element $0$ that satisfies $0<\tau^{\rm dl}_l<1/{\triangle f}$, and $\theta^{\rm dl}_l$ is the angle of the $l$th path in the downlink that satisfies $0<\theta^{\rm dl}_l<2\pi$.

Reciprocity does not normally apply in FDD systems because of the different operating frequencies in the uplink and downlink. Nonetheless, the uplink and downlink channels share a common propagation space between the BS and the user, and some partial reciprocity is expected if the frequency bands are within a certain coherent bandwidth.

The uplink and downlink signals propagate along common paths and are reflected by the same scatterers. Given that the wireless signals travel the same transmission distance and at the same speed, the delay is equal in both the uplink and downlink. According to the measurement results in \cite{Hugl2002} and \cite{Imtiaz2015}, the spatial directions or angles in the uplink channel are almost the same as those in the downlink channel. Therefore, $L^{\rm ul}=L^{\rm dl}=L$, $\tau^{\rm ul}_l=\tau^{\rm dl}_l=\tau_l$, and $\theta^{\rm ul}_l=\theta^{\rm dl}_l=\theta_l$ are obtained. The delays and angles $\{\tau_l, \theta_l \}_{l=0,\dots,L-1}$ are frequency-independent, thereby revealing a spatial reciprocity between the uplink and downlink.

\section{Efficient Downlink Channel Reconstruction}\label{Sec:Reconstruction}

The spatial reciprocity inspires us to reconstruct the FDD downlink channels by using the frequency-independent parameters estimated in the uplink instead of estimating the downlink CSI via massive downlink training and feedback. In this section, we propose an efficient downlink channel reconstruction scheme for FDD multi-antenna systems based on spatial reciprocity by using the uplink CSI with a small amount of downlink training and feedback overhead.

\subsection{Estimating Frequency-Independent Parameters}

During the uplink sounding process, the BS receives pilots sent from the user and is given the opportunity to estimate the frequency-independent parameters.
The pilot received by the BS antenna element $m$ on subcarrier $n$ can be expressed as
\begin{equation}\label{Eq:ULpilotsArray}
y^{\rm ul}_m (n) = h^{\rm ul}_m (n) s(n)+z^{\rm ul}_m(n)
= \sum\limits_{l = 0}^{{L-1}} {g^{\rm ul}_l e^{j2\pi n \triangle f\tau_l + j2\pi m\frac{d}{\lambda}\sin\theta_l}s(n)}+z^{\rm ul}_m(n),
\end{equation}
where $s(n)$ is the transmitted pilot on subcarrier $n$, and $z^{\rm ul}_m(n)$ is the complex Gaussian noise vector on BS antenna element $m$ and subcarrier $n$ with zero mean and unit variance.

In this multi-subcarrier multi-antenna system, the BS can receive pilots on each occupied subcarrier and antenna element. $N$ continuous subcarriers is assumed to be occupied by the pilots, while the central subcarrier is assumed  to be DC. The transmitted pilots on all subcarriers are equal to 1, thereby satisfying $s(-\lfloor {N/2}\rfloor )=\cdots=s(\lceil N/2\rceil-1)=1$. To detect the two-tuples $\{\tau_l, \theta_l \}_{l=0,\dots,L-1}$ from the received pilots, we stack the received pilots on all subcarriers and antennas together into a vector and obtain
\begin{equation}\label{Eq:ULpilots}
{\bf y}^{\rm ul} = \sum\limits_{l = 0}^{{L-1}}{g^{\rm ul}_l {\bf p}(\tau_l)\otimes{\bf a}(\theta_l)}+{\bf z}^{\rm ul},
\end{equation}
where ${\bf z}^{\rm ul}$ is the stacked uplink noise vector with i.i.d. elements. Here, we denote
\begin{equation}\label{Eq:uvec}
{\bf u}(\tau,\theta) = {\bf p}(\tau)\otimes{\bf a}(\theta).
\end{equation}
Therefore, \eqref{Eq:ULpilots} can be rewritten as
\begin{equation}\label{Eq:NewNOMPmodel}
{\bf y}^{\rm ul} = \sum\limits_{l = 0}^{{L-1}}{g^{\rm ul}_l {\bf u}(\tau_l,\theta_l)}+{\bf z}^{\rm ul}.
\end{equation}
Based on \eqref{Eq:NewNOMPmodel}, the parameter detection problem can be translated to a frequency detection problem.

From \cite{Mamandipoor2016}, we find that the NOMP algorithm behaves well in detecting frequencies from a mixture of sinusoids. NOMP can extract $(a_l, w_l)_{l=0,...,L-1}$ from
\begin{equation}\label{Eq:NOMPmodel}
{\bf y} = \sum\limits_{l = 0}^{{L-1}} {a_l {\bf x}(w_l)}+{\bf z},
\end{equation}
where ${\bf x}(w) = \left[1,e^{jw},\dots, e^{j(N-1)w}\right]^H$ and ${\bf z}$ is a complex Gaussian vector with i.i.d. elements. The NOMP algorithm introduced in \cite{Mamandipoor2016} only estimates one frequency parameter, namely, $w$. However, two frequency parameters need to be detected in our case, namely, $(\tau, \theta)$. In other words, three-tuples must be extracted from \eqref{Eq:NewNOMPmodel}, including $g^{\rm ul}$. The original NOMP algorithm cannot satisfy our requirement and we must extend it to fit the trivariate condition.

A detailed description of the trivariate NOMP algorithm is provided in \emph{Section \ref{Sec:NOMP}}. In the following part of this section, we suppose that the final results are obtained after the trivariate NOMP algorithm is applied to \eqref{Eq:NewNOMPmodel}. The detected three-tuples are recorded as $({\hat g}^{\rm ul}_l, {\hat \tau}_l, {\hat \theta}_l)_{l=0,...,\hat L-1}$.

\subsection{Necessity to Refine the Gains}

After obtaining the gain, delay, and angle of each path via uplink training and trivariate NOMP estimation, we reconstruct the downlink channel for the FDD transmission system.

As mentioned before, the delays and angles are frequency-independent parameters, and their uplink estimates can be applied in the reconstruction of the downlink frequency band channel. Based on these facts, \cite{Vasisht2016Eliminating} proposed an R2-F2 system that extracts the information of the propagation paths from the channels on band 1 in order to reconstruct the corresponding channels on band 2. This system allows the LTE BSs to infer the downlink channels by using the uplink-derived CSI and underscores the need to eliminate CSI feedback, which will significantly improve the time-frequency resource utilization. However, we are still unsure whether the reconstructed downlink channel is accurate enough if the uplink estimates are directly applied to the downlink channel model without using any downlink CSI.

In \cite{Vasisht2016Eliminating}, the precondition for eliminating downlink training and feedback is that all spatial parameters, including the gains, delays, and angles, are frequency-independent. The gain of each path is viewed to be identical in both the uplink and downlink. However, the existing measurements for the correlation of gains in different frequency bands do not provide sufficient evidence to confirm the frequency-independent feature of the gains. On the contrary, \cite{Imtiaz2015} demonstrated that due to phase difference, the power of a cluster differs in uplink and downlink. A more commonly accepted view is that the azimuth power spectrum, which can be regarded as the secondary moment of the complex gain, is highly correlated in both the uplink and downlink. This view has been confirmed by the measurements in \cite{Hugl2002}. In \cite{Xie2017}, the authors proposed to model the azimuth power spectrum based on a same shape but multiplied by a frequency-dependent factor, thereby suggesting that the amplitudes are not equal in different frequency bands. Previous studies generally hold that the instantaneous spatial complex gains are different in the uplink and downlink \cite{Goldberg1998,Liang2001}. Therefore, we cannot suppose that the gains are frequency-independent.

Meanwhile, even if the gains are assumed to be frequency-independent, the estimation errors will negatively affect the reconstructions on another frequency band. These errors are inevitable for any detection method, including NOMP and the optimization method used in \cite{Vasisht2016Eliminating}. The simplest single-antenna single-path case is used as an example to determine the impact of the estimation error. We denote the real gain and delay by $g$ and $\tau$, respectively, and denote the real channel on frequency $f_1$ by $h_1 = g e^{j2\pi f_1\tau}$. We assume that the estimated gain and delay on frequency $f_1$ are $\hat g$ and $\hat\tau$, respectively, and that the estimation error of delay is $\triangle\tau =\hat\tau -\tau$. For frequency $f_1$, $\hat g$ will compensate for the phase error caused by $\triangle\tau$ because the gain is updated by using LS estimation at the end of the NOMP algorithm. The reconstructed channel $\hat h_1 = \hat g e^{j2\pi f_1\hat \tau}$ will be very much the same as the original channel $h_1$. In this case, ``global accuracy'' is obtained instead of ``local accuracy''.

However, when $\hat g$ and $\hat\tau$ are used directly to reconstruct the channel on frequency $f_2$, the phase error $2\pi f_2 \triangle\tau$ is considerable if either $f_2$ or $\triangle\tau$ is large enough as shown in Fig.~\ref{Fig:PhaseError}(a). Meanwhile, the $\hat g$ derived on frequency $f_1$ is no longer able to compensate for this phase error on frequency $f_2$. The reconstructed channel on frequency $f_2$ is expressed as
\begin{equation}\label{Eq:EstimationError}
\hat h_2 = \hat g e^{j2\pi f_2\hat \tau} = h_2 \frac{\hat g}{g}e^{j2\pi f_2 \triangle\tau},
\end{equation}
where $\frac{\hat g}{g}e^{-j2\pi f_2 \triangle\tau}$ is the multiplicative estimation error from the global perspective. We can find that the derived channel on frequency $f_2$ is far from the real channel. Although $\hat g$ has the same absolute value as $g$, that is, $|\hat g| = |g|$, the phase difference between $\hat h_2$ and $h_2$ becomes unacceptably large because the phase information of the wireless channel is of great importance to the transceiver design.

\begin{figure}
  \centering
  \includegraphics[scale=0.8]{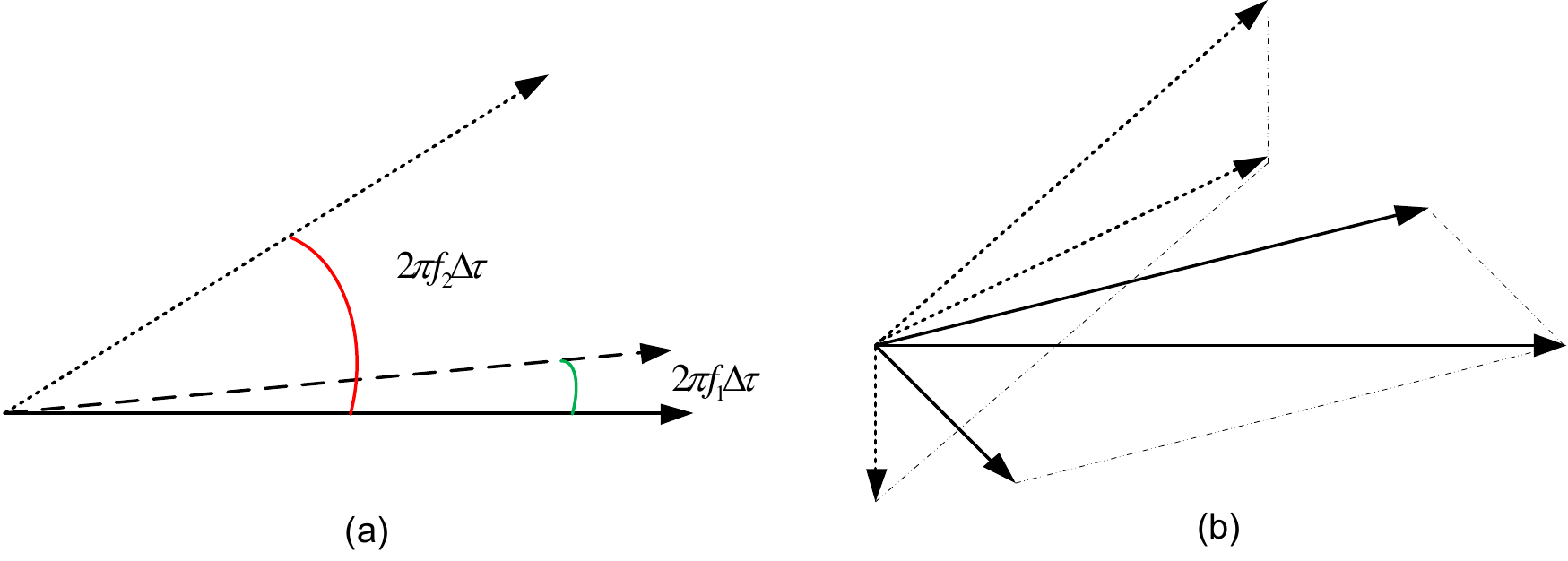}
  \caption{Impact of phase error: (a) Single-path scenario. The solid line with arrow represents the true phase of the path. The dotted lines with arrows are the phase errors on frequencies $f_1$ and $f_2$. (b) Multi-path scenario. The solid lines with arrows are the real multi-path components and the real superposition. The dotted lines with arrows are the estimated components and their superposition.}\label{Fig:PhaseError}
\end{figure}

The phase error will severely affect the multi-path channel reconstruction. As shown in Fig.~\ref{Fig:PhaseError}(b), the original channel comprises two paths with different amplitudes and delays, which are denoted by solid lines with arrows. As a result of the phase error, these multi-path components rotate and form an incorrect superposition. An angular error also takes place in highly complicated multi-antenna multi-path scenarios, thereby further harming the reconstruction on another band. Therefore, we do not suggest to follow the approach in \cite{Vasisht2016Eliminating}, which only uses the uplink CSI to reconstruct the downlink channel for FDD transmission systems.

\subsection{Re-Estimation and Reconstruction Scheme}

Given the inevitable estimation errors of the delays and angles, the gains are LS-estimated at the last step of the NOMP algorithm to compensate for these errors in the reconstruction of the uplink channel. Similarly, the gains can also be re-estimated via LS to compensate for the errors in reconstructing the downlink channel. This approach requires additional downlink overhead. Fortunately, only the gains need to be refined. Both the delays and angles estimated in the uplink are still applicable to the downlink channel reconstruction. Therefore, a small amount of overhead is required to refine the gains.

The gains are refined with the aid of pilots that are transmitted in the downlink. To retrieve the feature of the whole downlink frequency band, these pilots are sparsely distributed in the downlink band. We use comb-type all-one pilots and insert one pilot in every $K$ subcarriers. Afterward, $N_p=\lfloor N/K \rfloor$ subcarriers are occupied by the pilots, and the indices of the subcarriers are $n_0, n_2, \dots, n_{N_p-1}$.

Unlike the uplink, multiple antennas exist at the transmitter and a single antenna exists at the receiver. The pilots that are transmitted by the antenna array will be additively received at a single antenna. To enhance the received power, the pilots are beamformed before the transmission. Given the angles of the propagation paths $\{\hat\theta_l\}_{l=0,\dots,\hat L-1}$ that are estimated in the uplink, we target the pilots to these directions and concentrate the transmit power onto the propagation path of the channel. The following beamforming types are considered here:

{\em Type 1:} The pilots in one OFDM symbol are beamformed to target one specific direction. We need $\hat L$ OFDM symbols to send the pilots, and different OFDM symbols correspond to different directions. For subcarrier $n_i$ on the $j$th OFDM symbol, the received pilot can be expressed as
\begin{equation}\label{Eq:DLpilot1}
y^{\rm dl}_j(n_i) = \sum\limits_{l = 0}^{{L-1}} {g^{\rm dl}_l e^{j2\pi (\triangle F+n_i \triangle f) \tau_l} {\bf a}^H(\theta_l)} {\bf a}(\hat\theta_j) + z^{\rm dl}_j(n_i),
\end{equation}
where $z^{\rm dl}_j(n_i)$ is the downlink noise on subcarrier $n_i$ and OFDM symbol $j$, $i=0,\dots,N_p-1$, and $j = 0,\dots,\hat L-1$.

{\em Type 2:} The pilots are frequency-division multiplexed onto different directions, and only one OFDM symbol is needed. Subcarriers $n_0,\dots,n_{\hat L-1}$ correspond to directions $\hat\theta_0,\dots,\hat\theta_{\hat L-1}$, respectively. The received pilot on subcarrier $n_i$ can be expressed as
\begin{equation}\label{Eq:DLpilot2}
y^{\rm dl}(n_i) = \sum\limits_{l = 0}^{{L-1}} {g^{\rm dl}_l e^{j2\pi (\triangle F+n_i \triangle f) \tau_l} {\bf a}^H(\theta_l)} {\bf a}(\hat\theta_{\underline i}) + z^{\rm dl}(n_i),
\end{equation}
where $\underline i=i\!\!\!\mod\hat L$, $z^{\rm dl}(n_i)$ is the downlink noise on subcarrier $n_i$, and $i = 0,\dots,N_p-1$.

To estimate the downlink gains at the user side, the BS needs to inform the user with the uplink estimated parameters $(\hat\tau_l,\hat\theta_l),{l=0,\dots,\hat L-1}$ and the beamforming type. The user applies these estimates into \eqref{Eq:DLpilot1} or \eqref{Eq:DLpilot2} according to the beamforming type and rewrites the signal models as
\begin{equation}\label{Eq:DLpilot1m}
y^{\rm dl}_j(n_i) = \sum\limits_{l = 0}^{{\hat L-1}} {g^{\rm dl}_l e^{j2\pi (\triangle F+n_i \triangle f) \hat\tau_l} {\bf a}^H(\hat\theta_l)} {\bf a}(\hat\theta_j) + z^{\rm dl}_j(n_i)
\end{equation}
and
\begin{equation}\label{Eq:DLpilot2m}
y^{\rm dl}(n_i) = \sum\limits_{l = 0}^{{\hat L-1}} {g^{\rm dl}_l e^{j2\pi (\triangle F+n_i \triangle f) \hat \tau_l} {\bf a}^H(\hat \theta_l)} {\bf a}(\hat\theta_{\underline i}) + z^{\rm dl}(n_i).
\end{equation}

After stacking the received pilots on all subcarriers and OFDM symbols, the following unified signal model for both types is constructed as
\begin{equation}\label{Eq:DLpilots}
{\bf y}^{\rm dl} = {\bf Ag}^{\rm dl}+{\bf z}^{\rm dl},
\end{equation}
where ${\bf y}^{\rm dl}$ and ${\bf z}^{\rm dl}$ are the stacked $M_p\times 1$ dimensional received pilots, gains, and noise vectors, respectively, while ${\bf A}$ denotes the $M_p\times \hat L$ dimensional coefficient matrix.

For Type 1, $M_p=N_p\hat L$, and the matrix ${\bf A}$ comprises $\hat L$ submatrices
\begin{equation}\label{Eq:matrixA1}
{\bf A} = \left[\begin{matrix} {\bf A}^{(0)} \\ \vdots \\ {\bf A}^{(\hat L-1)} \end{matrix} \right]^T,
\end{equation}
where the $(i,l)$th entry of the submatrix ${\bf A}^{(j)}$ is equal to
\begin{equation}\label{Eq:matrixA1elmt}
{\bf A}^{(j)}_{i,l} = e^{j2\pi (\triangle F+n_i \triangle f) \hat\tau_l} {\bf a}^H(\hat\theta_l) {\bf a}(\hat\theta_j),
\end{equation}
where $j = 0,\dots,\hat L-1$, $i = 0,\dots,N_p-1$ and $l = 0,\dots,\hat L-1$.

For Type 2, $M_p=N_p$, and the $(i,l)$th entry of $A$ is equal to
\begin{equation}\label{Eq:matrixA1elmt}
{\bf A}_{i,l} = e^{j2\pi (\triangle F+n_i \triangle f) \hat\tau_l} {\bf a}^H(\hat\theta_l) {\bf a}(\hat\theta_{\underline i}),
\end{equation}
where $i = 0,\dots,N_p-1$, and $l = 0,\dots,\hat L-1$.

Given that the coefficient matrix $\bf A$ is also known at the user side, the user can refine the gains via LS estimation as
\begin{equation}\label{Eq:DLestg}
{\hat {\bf g}}^{\rm dl}= {\bf A}^{\dagger}{\bf y}^{\rm dl},
\end{equation}
where ${\bf A}^{\dagger}$ represents the pseudo-inverse of ${\bf A}$, and the dimension of the refined gain vector ${\hat {\bf g}}^{\rm dl}$ is still $\hat L$. The refined gains are then sent back to the BS. The feedback amount is independent of the number of antenna elements and subcarriers but is dependent on the number of detected propagation paths.

The BS obtains all the information required for the reconstruction of the downlink channel, namely, $({\hat g}^{\rm dl}_l,\hat\tau_l,\hat\theta_l),{l=0,\dots,\hat L-1}$. Specifically, the downlink multipath channel on all subcarriers and antennas is reconstructed as
\begin{equation}\label{Eq:Reconstruction}
{\bf h}^{\rm dl}=\sum\limits_{l = 0}^{{\hat L-1}} {{\hat g}^{\rm dl}_l e^{j2\pi \triangle F \hat\tau_l} {\bf p}(\hat\tau_l) \otimes {\bf a}(\hat\theta_l)}.
\end{equation}

For clarity, we briefly summarize the procedures used in the proposed efficient downlink channel reconstruction scheme as follows:
\begin{itemize}
  \item \textbf{Step 1:} Frequency-independent parameters estimation during the uplink sounding. The user sends uplink pilots to the BS and then the BS uses the extended triviriate NOMP algorithm to estimate the gain, delay, and angle of each propagation path of the channel.
  \item \textbf{Step 2:} Downlink gain refinement and feedback. The BS transmits the downlink pilots to the user and informs the user about the beamforming type and the uplink-estimated delays and angles. The user re-estimates the gains and then feeds them back to the BS.
  \item \textbf{Step 3:} Downlink channel reconstruction. The BS reconstructs the downlink channel by using the uplink-estimated delays and angles as well as the downlink-refined gains.
\end{itemize}

\section{Uplink Parameters Extraction}\label{Sec:NOMP}

In this section, we describe the extension of the NOMP algorithm in detail in order to fit the trivariate condition and obtain the three-tuple $(g^{\rm ul},\tau,\theta)$ from the uplink as mentioned in \emph{Section \ref{Sec:Reconstruction}.A}. We first introduce the rationale and stopping criterion of the trivariate NOMP algorithm and then evaluate its accuracy by using the lower bounds of the estimation errors.

For simplification, $\triangle f\tau$ is treated as a whole that is simplified by $\mu\in [0,1)$. Similarly, $\frac{d}{\lambda}\sin\theta$ is simplified by $\nu \in [0,1)$. Then, vector $\bf u$ is represented by
\begin{equation}\label{Eq:uvecNew}
{\bf u}(\mu,\nu) = {\bf p}(\mu)\otimes{\bf a}(\nu)
\end{equation}
in the subsequent part of this section, where
\begin{equation}\label{Eq:pvecNew}
{\bf p}(\mu) = \left[e^{-j2\pi \lfloor\frac{N}{2}\rfloor \mu},\dots, e^{j2\pi (\lceil\frac{N}{2}\rceil-1) \mu}\right]^H
\end{equation}
and
\begin{equation}\label{Eq:avecNew}
{\bf a}(\nu) = \left[e^{-j2\pi \lfloor\frac{M}{2}\rfloor \nu},\dots, e^{j2\pi (\lceil\frac{M}{2}\rceil-1) \nu}\right]^H.
\end{equation}
The three-tuples to be detected are transformed to $(g^{\rm ul},\mu,\nu)$ in the extended trivariate NOMP algorithm.

\subsection{Trivariate NOMP Algorithm}\label{Sec:Extended Algorithm}

NOMP is an iteration-based algorithm. In our extended version of this algorithm, a three-tuple of $(g^{\rm ul},\mu,\nu)$ is estimated in each iteration. The component made by this three-tuple is then removed from the observed pilot. At the end of the $i$th iteration, the residual is calculated as
\begin{equation}\label{Eq:NOMPresidual}
{\bf y}^{\rm ul}_{\rm r} = {\bf y}^{\rm ul} - \sum\limits_{l = 0}^{{i-1}}{{\tilde g}^{\rm ul}_l {\bf u}(\tilde\mu_l,\tilde\nu_l)},
\end{equation}
where $({\tilde g}^{\rm ul}_l,\tilde\mu_l,\tilde\nu_l),l=0,\dots,i-1$ are the estimated three-tuples in the previous $i$ iterations. Afterward, in the $(i+1)$th iteration, we estimate a new three-tuple by minimizing the new residual power $\|{\bf y}^{\rm ul}_{\rm r} - g^{\rm ul}{\bf u}(\mu,\nu)\|^2$, which is further translated to maximize the following function
\begin{equation}\label{Eq:Sfunction}
S(g^{\rm ul},\mu,\nu) = 2\mathcal{R}\left\{{\bf y}^{{\rm ul}H}_{\rm r} g^{\rm ul}{\bf u}(\mu,\nu)\right\} - |g^{\rm ul}|^2\|{\bf u}(\mu,\nu)\|^2.
\end{equation}

The working steps in the $(i+1)$th iteration of the extended trivariate NOMP algorithm are similar to those of the original algorithm in \cite{Mamandipoor2016}. We first briefly introduce the steps in the $(i+1)$th iteration of the extended algorithm, which are listed below.
\begin{itemize}
  \item Step 1: New Detection. Select the coarse estimates $\tilde\mu_i$ and $\tilde\nu_i$ from a 2D over-sampled angle-and-delay grid and then calculate ${\tilde g}^{\rm ul}_i$ from $\tilde\mu_i$ and $\tilde\nu_i$.
  \item Step 2: Single Refinement. Solely refine the coarsely estimated three-tuple $({\tilde g}^{\rm ul}_i,\tilde\mu_i,\tilde\nu_i)$ through the Newton refinement steps and add the obtained $({{\tilde g}^{\rm ul}_i}{'}, {\tilde\mu_i}{'}, {\tilde\nu_i}{'})$ into the set of the estimated three-tuples.
  \item Step 3: Cyclic Refinement. Cyclically refine the set of estimated three-tuples through the Newton refinement steps and obtain $({\tilde g}^{\rm ul}_l{''},\tilde\mu_l{''},\tilde\nu_l{''}), l=0,\dots,i$.
  \item Step 4: Gains Update. Retain the estimated delays and angles and update all the amplitudes through LS estimation $[{\tilde{\tilde g}}^{\rm ul}_0,{\tilde{\tilde g}}^{\rm ul}_1,\dots,{\tilde{\tilde g}}^{\rm ul}_i]^T = {\bf U}^\dagger {\bf y}^{\rm ul}$, where ${\bf U} = \left[{\bf u}(\tilde\mu_0{''},\tilde\nu_0{''}), {\bf u}(\tilde\mu_1{''},\tilde\nu_1{''}), \dots ,{\bf u}(\tilde\mu_i{''},\tilde\nu_i{''})\right]$.
\end{itemize}

Details about the extensions of this work are then provided. Given that the delay and the angle jointly determine the channel phase and can be represented using a common vector ${\bf u}(\mu,\nu)$, these two parameters are estimated and refined together in our design. This combination results in the 2D grid and the extended Newton step.

\subsubsection{2D Grid}

The coarse estimates $(\tilde\mu_i,\tilde\nu_i)$ in Step 1 are chosen from a 2D angle-and-delay grid $\Omega$, which consists of $\gamma_1 N \times\gamma_2 M$ over-sampled grid points
\begin{equation}\label{Eq:Grid}
\Omega = \left\{\left(\frac{k_1}{\gamma_1 N}, \frac{k_2}{\gamma_2 M}\right): k_1 = 0,1,\dots,{\gamma_1 N-1};
k_2 = 0,1,\dots,{\gamma_2 M-1}\right\},
\end{equation}
where $\gamma_1$ and $\gamma_2$ are the over-sampling rates for the delay grid and the angle grid, respectively. Each point in the grid forms a vector ${\bf u}({k_1}/{(\gamma_1 N)}, {k_2}/{(\gamma_2 M)})$. The coarsely estimated delay and angle are obtained by exhaustively searching the grid points as follows:
\begin{equation}\label{Eq:DetectNew}
(\tilde\mu_i,\tilde\nu_i) = \arg\max_{(\mu,\nu)\in\Omega}{\frac{|{\bf u}^H(\mu,\nu){\bf y}^{\rm ul}_{\rm r}|^2}{\|{\bf u}(\mu,\nu)\|^2}}.
\end{equation}
Next, the gain is calculated as
\begin{equation}\label{Eq:gcal}
{\tilde g}^{\rm ul}_i = \frac{{\bf u}^H(\tilde\mu_i,\tilde\nu_i){\bf y}^{\rm ul}_{\rm r}}{\|{\bf u}(\tilde\mu_i,\tilde\nu_i)\|^2}.
\end{equation}

\subsubsection{Extended Newton Step}

With one more parameter than the original Newton step, the extended Newton step in Steps 2 and 3 can refine the delay and the angle simultaneously. In this bivariate problem, the coarsely estimated $(\tilde\mu_i,\tilde\nu_i)$ are refined through
\begin{equation}\label{Eq:RefineOne}
\left[ \begin{matrix} {\tilde\mu'}_i\\ {\tilde\nu'}_i \end{matrix} \right]=\left[ \begin{matrix} {\tilde\mu_i}\\ {\tilde\nu_i} \end{matrix} \right]- {\ddot{\bf S}\left({\tilde g}^{\rm ul}_i,\tilde\mu_i,\tilde\nu_i\right)}^{-1} {\dot{\bf S}\left({\tilde g}^{\rm ul}_i,\tilde\mu_i,\tilde\nu_i\right)},
\end{equation}
where
\begin{equation}\label{Eq:NOMPdotS}
\dot{\bf S}\left(g^{\rm ul},\mu,\nu\right)=\left[\begin{matrix} {\frac{\partial S}{\partial\mu}} \\ \frac{\partial S}{\partial\nu} \end{matrix} \right]
\end{equation}
is the first-order partial derivative vector, and
\begin{equation}\label{Eq:NOMPddotS}
\ddot{\bf S}\left(g^{\rm ul},\mu,\nu\right)=\left[ \begin{matrix} \frac{\partial^2 S}{\partial\mu^2} & \frac{\partial^2 S}{\partial\mu\partial\nu}\\ \frac{\partial^2 S}{\partial\nu\partial\mu} & \frac{\partial^2 S}{\partial\nu^2} \end{matrix} \right]
\end{equation}
is the second-order partial derivative matrix. According to \eqref{Eq:Sfunction}, we can write the first-order partial derivatives of $S(g^{\rm ul},\mu,\nu)$ as
\begin{equation}\label{Eq:1orderpartialS}
\frac{\partial S}{\partial x} = 2\mathcal{R}\left\{g^{\rm ul}\left({\bf y}^{\rm ul}_{\rm r}- g^{\rm ul}{\bf u}\right)^H \frac{\partial {\bf u}}{\partial x} \right\} ,
\end{equation}
where $x$ can be $\mu$ or $\nu$. The second-order partial derivative of $S(g^{\rm ul},\tau,\theta)$ is calculated as
\begin{equation}\label{Eq:mixedpartialS}
\frac{\partial^2 S}{\partial x_1 \partial x_2}= 2\mathcal{R}\left\{g^{\rm ul}\left({\bf y}^{\rm ul}_{\rm r}- g^{\rm ul}{\bf u}\right)^H \frac{\partial^2 {\bf u}}{\partial x_1 \partial x_2} - |g^{\rm ul}|^2 \frac{\partial {\bf u}^H}{\partial x_2} \frac{\partial {\bf u}}{\partial x_1} \right\},
\end{equation}
where $x_1$ and $x_2$ can be $\mu$ or $\nu$.
One requirement is that $S(g^{\rm ul},\mu,\nu)$ is locally convex in the neighborhood of $(\tilde\mu,\tilde\nu)$ because we are pursuing its maximum value. Therefore, the Newton refinement \eqref{Eq:RefineOne} will be carried out if, and only if, $\det\left(\ddot{\bf S}\left({\tilde g}^{\rm ul}, \tilde\mu,\tilde\nu\right)\right)>0$ and the first element of $\ddot{\bf S}\left({\tilde g}^{\rm ul},\tilde\mu,\tilde\nu\right)$ is lower than 0. At the end of each Newton step, the gain is also updated using \eqref{Eq:gcal}.

Note that the 2D grid and the extended Newton step are the required major extensions to the original NOMP algorithm. Other minor modifications to fit the trivariate condition are trivial and omitted here.

\subsection{Stopping Criterion}

One major challenge is that BS does not know the number of propagation paths in the real channel, a detail which directly determines when the iteration process is terminated. If the estimated three-tuples are precise enough, all the paths will be accurately identified and the residual will be reduced to the noise at the end of the NOMP algorithm, i.e., ${\bf y}^{\rm ul}_{\rm r}\approx{\bf z}^{\rm ul}$. In this study, this assumption is utilized to design the stopping criterion.

\subsubsection{Power-based Criterion}
One choice is to terminate the NOMP iterations when the residual power is less than the total noise power. Since the noise power is normalized to 1, if
\begin{equation}\label{Eq:Stop1}
\|{\bf y}^{\rm ul}_{\rm r} \|^2 < \kappa,
\end{equation}
where
\begin{equation}\label{Eq:StopCriterion1}
\kappa = \mathbb{E}\{\|{\bf z}^{\rm ul}\|^2 \}= M N,
\end{equation}
then the trivariate NOMP algorithm will be stopped.

\subsubsection{False-Alarm-Rate-based Criterion}
Alternatively, we can design the stopping criterion based on the false alarm rate. If we ``detect'' a fake path that does not exist, then we say that a fault appears. This situation happens when all the paths have been detected but the algorithm is still not working. The following theorem introduces the false-alarm-rate-based stopping criterion.

\begin{theorem}\label{Thrm:criterion}
If the trivariate NOMP algorithm terminates when
\begin{equation}\label{Eq:Stop2}
\|{\bf u}(\mu,\nu)^H{\bf y}^{\rm ul}_{\rm r}\|^2 < \kappa'
\end{equation}
holds for all grid points
\begin{equation}\label{Eq:DFTgrids}
(\mu,\nu)\in\left\{\left(\frac{k_1{'}}{N}, \frac{k_2{'}}{M}\right): k_1{'} = 0,1,\dots,{N-1};
k_2{'} = 0,1,\dots,{M-1}\right\},
\end{equation}
where
\begin{equation}\label{Eq:StopCriterion2}
\kappa' = \ln(M N)-\ln(-\ln(1-P_{\rm fa})),
\end{equation}
then the false alarm rate can be approximated by $P_{\rm fa}$.
\end{theorem}

\begin{proof}
As the grid points listed in \eqref{Eq:DFTgrids} are non-over-sampled points, the corresponding values of ${\bf u}(\mu,\nu)^H{\bf y}^{\rm ul}_{\rm r}$ can be viewed as the Fourier transformed values of ${\bf y}^{\rm ul}_{\rm r}$ and remain the same statistic property of ${\bf y}^{\rm ul}_{\rm r}$. Since ${\bf y}^{\rm ul}_{\rm r}\approx{\bf z}^{\rm ul}$ when all the paths are precisely detected, the condition \eqref{Eq:Stop2} can be translated to
\begin{equation}\label{Eq:ThrmProof1}
\|{\bf z}^{\rm ul}\|^2_\infty < \kappa'.
\end{equation}
From \cite{Eisenberg2008}, we know that
\begin{equation}\label{Eq:maximum}
\mathbb{E}\left\{\|{\bf z}^{\rm ul}\|^2_\infty\right\} \approx \ln(MN)
\end{equation}
when $MN$ grows large. Denoting $\bar z=\|{\bf z}^{\rm ul}\|^2_\infty - \ln(MN)$, we can derive that
\begin{equation}\label{Eq:xcdf}
\mathcal{P}\left\{\bar z \le Z \right\} = \mathcal{P}\left\{\|{\bf z}^{\rm ul}\|^2_\infty < Z+\ln(MN) \right\}.
\end{equation}
Given that each element of $\bf z$ is i.i.d., it holds that
\begin{equation}\label{Eq:iidGaussian}
\mathcal{P}\left\{\|{\bf z}^{\rm ul}\|^2_\infty < Z+\ln(MN) \right\} = \mathcal{P}\left\{|z_1|^2 \le Z+\ln(MN)\right\}^{MN}=\left(1-\frac{1}{MN}e^{-Z}\right)^{MN},
\end{equation}
where $z_1$ is the first element of ${\bf z}^{\rm ul}$ and is a Gaussian variable with zero mean and unit variance. As
\begin{equation}\label{Eq:explimit}
\exp\{\xi\} = \lim_{n \to \infty}{\left(1+\frac{\xi}{n}\right)^n},
\end{equation}
we have
\begin{equation}\label{Eq:elimit}
\left(1-\frac{1}{MN}e^{-Z}\right)^{MN}\approx \exp\{-\exp\{-Z\}\}
\end{equation}
when $MN$ grows without limit. By applying \eqref{Eq:iidGaussian}--\eqref{Eq:elimit} into \eqref{Eq:xcdf} and denoting $\kappa'=Z+\ln(MN)$, we can obtain
\begin{equation}\label{Eq:zmax}
\mathcal{P}\left\{\|{\bf z}^{\rm ul}\|^2_\infty \le \kappa' \right\}\approx \exp\{-\exp\{\ln(MN)-\kappa'\}\}.
\end{equation}
If we further apply \eqref{Eq:StopCriterion2}, then it holds that
\begin{equation}\label{Eq:zmax}
\mathcal{P}\left\{\|{\bf z}^{\rm ul}\|^2_\infty > \kappa' \right\}=1-\mathcal{P}\left\{\|{\bf z}^{\rm ul}\|^2_\infty \le \kappa' \right\}\approx P_{\rm fa},
\end{equation}
which means that the false alarm rate approximates $P_{\rm fa}$.
\end{proof}

After the iterations stop, the final estimation results of the trivariate NOMP algorithm are denoted as $({\hat g^{\rm ul}_l}, {\hat\mu_l}, {\hat\nu_l}), l=0,\dots,\hat L-1$. These estimation results are further translated to $({\hat g^{\rm ul}_l}, {\hat\tau_l}, {\hat\theta_l}), l=0,\dots,\hat L-1$, as mentioned in \emph{Section \ref{Sec:Reconstruction}}.

\subsection{Estimation Accuracy}

To evaluate the estimation accuracy of the delay and the angle, we calculate the respective normalized mean square errors (MSEs) of $\mu$ and $\nu$ by
\begin{equation}\label{Eq:delayError}
\varepsilon_\mu = \mathbb{E}\left\{ \frac{|\hat\mu-\mu|^2}{1/{N^2}} \right\}
\end{equation}
and
\begin{equation}\label{Eq:angleError}
\varepsilon_\nu = \mathbb{E}\left\{ \frac{|\hat\nu-\nu|^2}{1/{M^2}} \right\}.
\end{equation}
The following theorem is used to study the estimation accuracy of the extended trivariate NOMP algorithm.

\begin{theorem}\label{Thrm:CRB}
The normalized MSEs of the delay and angle are lower bounded, respectively, by
\begin{equation}\label{Eq:delayTheorem}
\varepsilon_\mu \ge \frac{3N}{{\rm SNR}\cdot2\pi^2 M (N^2-1)}
\end{equation}
and
\begin{equation}\label{Eq:angleTheorem}
\varepsilon_\nu \ge \frac{3M}{{\rm SNR}\cdot2\pi^2 N (M^2-1)}.
\end{equation}
\end{theorem}

\begin{proof}
Cramer-Rao bound (CRB) can be interpreted as a lower bound of the variance of the estimator. The CRBs of the single path case, ${\bf y}^{\rm ul} = g^{\rm ul}{\bf u}(\mu,\nu)+{\bf z}^{\rm ul}$, are introduced, where each element of ${\bf z}$ is i.i.d Gaussian with zero mean and unit variance. According to \cite{Ramasamy2014}, the Fisher information matrix is calculated by
\begin{equation}\label{Eq:FisherDefinition}
{\bf F}(\mu,\nu) = 2|g^{\rm ul}|^2\mathcal{R}\left\{ \left[ \begin{matrix}
\frac{\partial {\bf u}^H}{\partial\mu} \frac{\partial {\bf u}}{\partial\mu}
& \frac{\partial {\bf u}^H}{\partial\mu} \frac{\partial {\bf u}}{\partial\nu} \\
\frac{\partial {\bf u}^H}{\partial\nu} \frac{\partial {\bf u}}{\partial\mu}
& \frac{\partial {\bf u}^H}{\partial\nu} \frac{\partial {\bf u}}{\partial\nu} \end{matrix} \right] \right\}.
\end{equation}
Applying \eqref{Eq:uvec} into \eqref{Eq:FisherDefinition}, we can get an analytical expression of the Fisher information matrix as
\begin{equation}\label{Eq:Fisher}
{\bf F}(\mu,\nu) = 2|g^{\rm ul}|^2 \left[ \begin{matrix} \frac{\pi^2 M N(M^2-1)}{3} & 0 \\ 0 & \frac{\pi^2 M N (N^2-1)}{3} \end{matrix} \right].
\end{equation}
Then, the CRB of the delay is expressed as
\begin{equation}\label{Eq:CRBdelay}
{\rm CRB}_\mu = {{\bf F}_{1,1}^{-1}(\mu,\nu)} = \frac{3}{2|g^{\rm ul}|^2\pi^2 MN (N^2-1)}.
\end{equation}
Similarly, the CRB of the angle is
\begin{equation}\label{Eq:CRBaoa}
{\rm CRB}_\nu = {{\bf F}_{2,2}^{-1}(\mu,\nu)} = \frac{3}{2|g^{\rm ul}|^2\pi^2 MN (M^2-1)}.
\end{equation}
With $a$ containing the pilot and the noise power equaling 1, the signal-to-noise ratio (SNR) here is measured through $|g^{\rm ul}|^2$, that is, $|g^{\rm ul}|^2={\rm SNR}$. Moreover, ${\rm CRB}_\mu \le \mathbb{E}\left\{{|\hat\mu-\mu|^2} \right\}$ and ${\rm CRB}_\nu \le \mathbb{E}\left\{{|\hat\nu-\nu|^2}\right\}$, from where we obtain \eqref{Eq:delayTheorem} and \eqref{Eq:angleTheorem}, respectively.
\end{proof}

When we set $M=1$, the problem is reduced to the bivariate case that only the gain and the delay are to be estimated. If ${\bf p}(\mu)=\left[e^{-j\lfloor{N/2}\rfloor \mu},\dots, e^{j (\lceil{N/2}\rceil-1) \mu}\right]^H/{\sqrt N}$, then the CRB of delay is written as
\begin{equation}\label{Eq:CRBdelay1}
{\rm CRB}_\mu^{(M=1)} = \frac{6}{{\rm SNR} (N^2-1)},
\end{equation}
which is exactly in accordance with the CRB bound given in \cite{Mamandipoor2016}. It proves the correctness of \emph{Theorem \ref{Thrm:CRB}}.

\begin{corollary}\label{Cor:CRB}
When $N$ or $M$ grows large, the lower bounds of the normalized MSEs of the delay and angle coincide, i.e.,
\begin{equation}\label{Eq:CRBCorollary}
\varepsilon_\mu,\varepsilon_\nu \ge \frac{3}{{\rm SNR}\cdot2\pi^2 MN}.
\end{equation}
\end{corollary}

\begin{proof}
It holds that $N^2/{(N^2-1)} \approx 1$ when $N$ grows large. Then \eqref{Eq:delayTheorem} approaches \eqref{Eq:CRBCorollary}. Similarly, \eqref{Eq:angleTheorem} approaches \eqref{Eq:CRBCorollary} when $M$ grows large.
\end{proof}

{\em Remark}: From \emph{Theorem \ref{Thrm:CRB}} and \emph{Corollary \ref{Cor:CRB}}, we can find that the bounds can be further lowered if the number of subcarriers occupied by the pilots or the number of BS antenna elements increases. This is because with more observed samples, we can see more details about the spatial channel. What should be emphasized are the preconditions of high estimation accuracy, i.e., the angles and delays of different paths are well separated and that the number of paths is far less than $M$ or $N$. In addition, only if the channel satisfies these preconditions can the algorithm achieve lower-bound performances.

\begin{figure}
  \centering
  \includegraphics[scale=0.5]{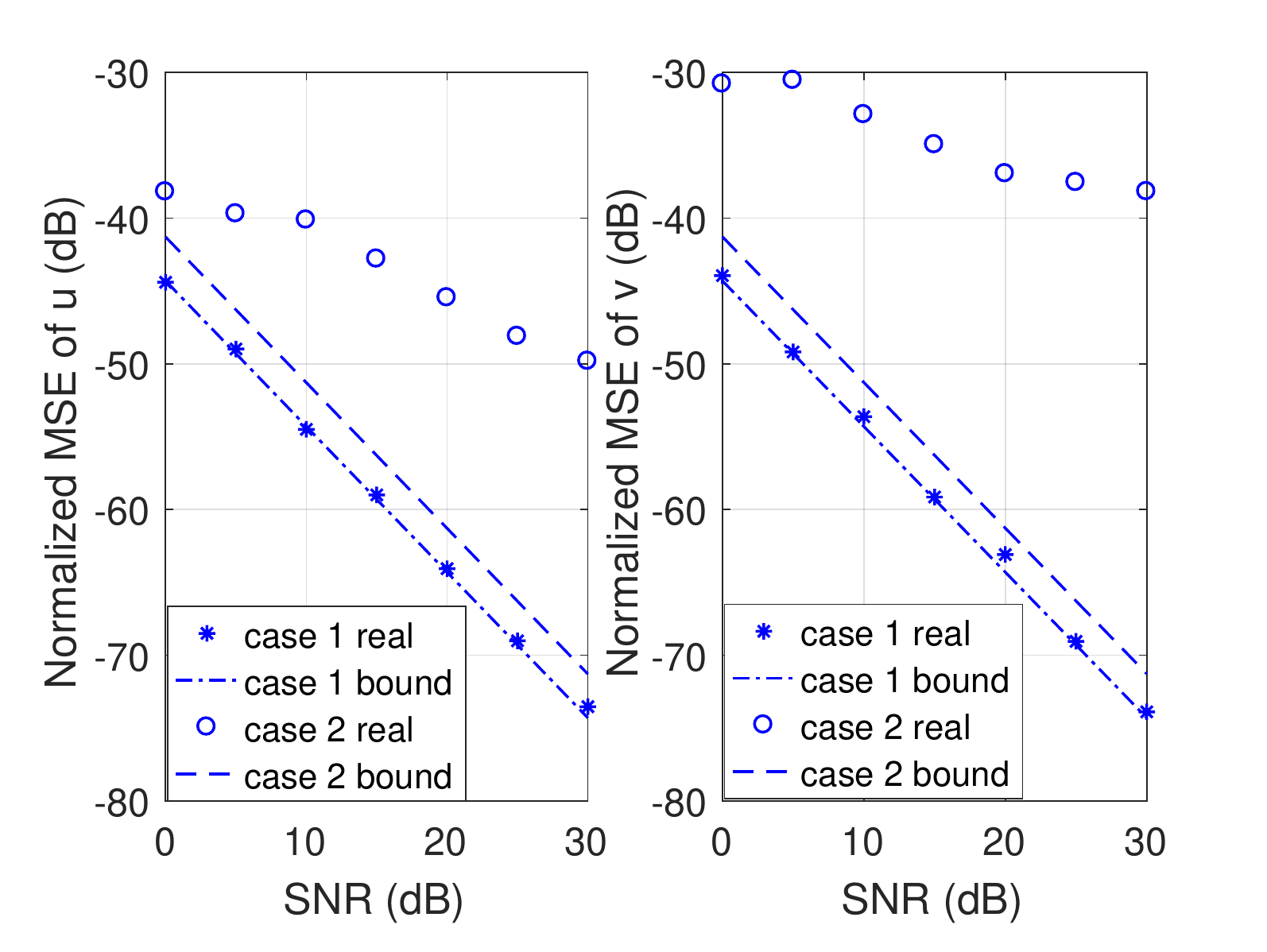}
  \caption{Comparison of the practical normalized MSEs and the lower bounds for $\mu$ and $\nu$. Case 1: $M=32$, $N=128$; Case 2: $M=32$, $N=64$.} \label{Fig:CRBresults}
\end{figure}

The results in Fig.~\ref{Fig:CRBresults} provide an intuitive comparison of the normalized MSEs with the derived lower bounds. The circles and stars represent the practical MSEs of the trivariate NOMP algorithm, and the dotted lines are the lower bounds. We set $\gamma_1=\gamma_2=2$. A total of 15 equal-power paths are present in the channel. The minimum separations among the delays and the angles are no less than $1/{N}$ and $1/M$, respectively. We first evaluate case 1, where $M=32$ and $N=128$. The values of $M$ and $N$ satisfy the condition in {\em Corollary \ref{Cor:CRB}} and we find that the lower bounds of the delay and the angle are nearly the same. Besides, the estimation accuracy is enhanced proportionally with the increase of SNR. The practical MSEs closely coincide with their theoretical lower bounds for both $\mu$ and $\nu$, which demonstrates the high accuracy of the trivariate NOMP algorithm. Moreover, even though $M \ll N$, the practical estimation accuracy of $\nu$ is not inferior to that of $\mu$ because of the well-separated spatial angles of each path and the significantly lower number of paths compared to $M$ or $N$. The results in Fig.~\ref{Fig:CRBresults} also compare the performances of the algorithm in case 2, where $M=32$ and $N=64$, that is, the number of subcarriers is half of that in case 1. The results demonstrate that the practical estimation accuracy degrades and the MSE lines deviate with the theoretical lower bounds. The lower bounds are accessed when the observations are far more than the propagation paths. Despite this, the MSEs of $\mu$ and $\nu$ are below $-30$ dB, demonstrating that the practical estimation accuracy of the delay and angle remain high.

\section{Performance Evaluation}\label{Sec:Results}

In this section, we evaluate the performance of the proposed efficient downlink channel reconstruction scheme. We first discuss our computer simulation results and then move on to our hardware OTA tests for validation.

\subsection{Simulation Results}\label{Sec:MatlabResults}

Computer simulations are realized through MATLAB. For the NOMP algorithm, the over-sampling rates of the delay and angle are set to 2 and 4, respectively. One round of single refinement and three rounds of cyclic refinement are implemented during each NOMP iteration. The number of FFT points is set as 2048 and the subcarrier spacing is set as 75 kHz. Note that we infer the out-of-band or downlink channel solely using the in-band or uplink derived gains, delays and angles as suggested in \cite{Vasisht2016Eliminating}, while we reconstruct the channel by utilizing the out-of-band or downlink refined gains and in-band or uplink derived delays and angles as suggested by the proposed efficient reconstruction scheme. Channel inference is equivalent to channel reconstruction for the in-band or the uplink, and both are realized by applying the NOMP estimated gains, delays, and angles in the channel model.

\begin{figure}
  \centering
  \includegraphics[scale=0.22]{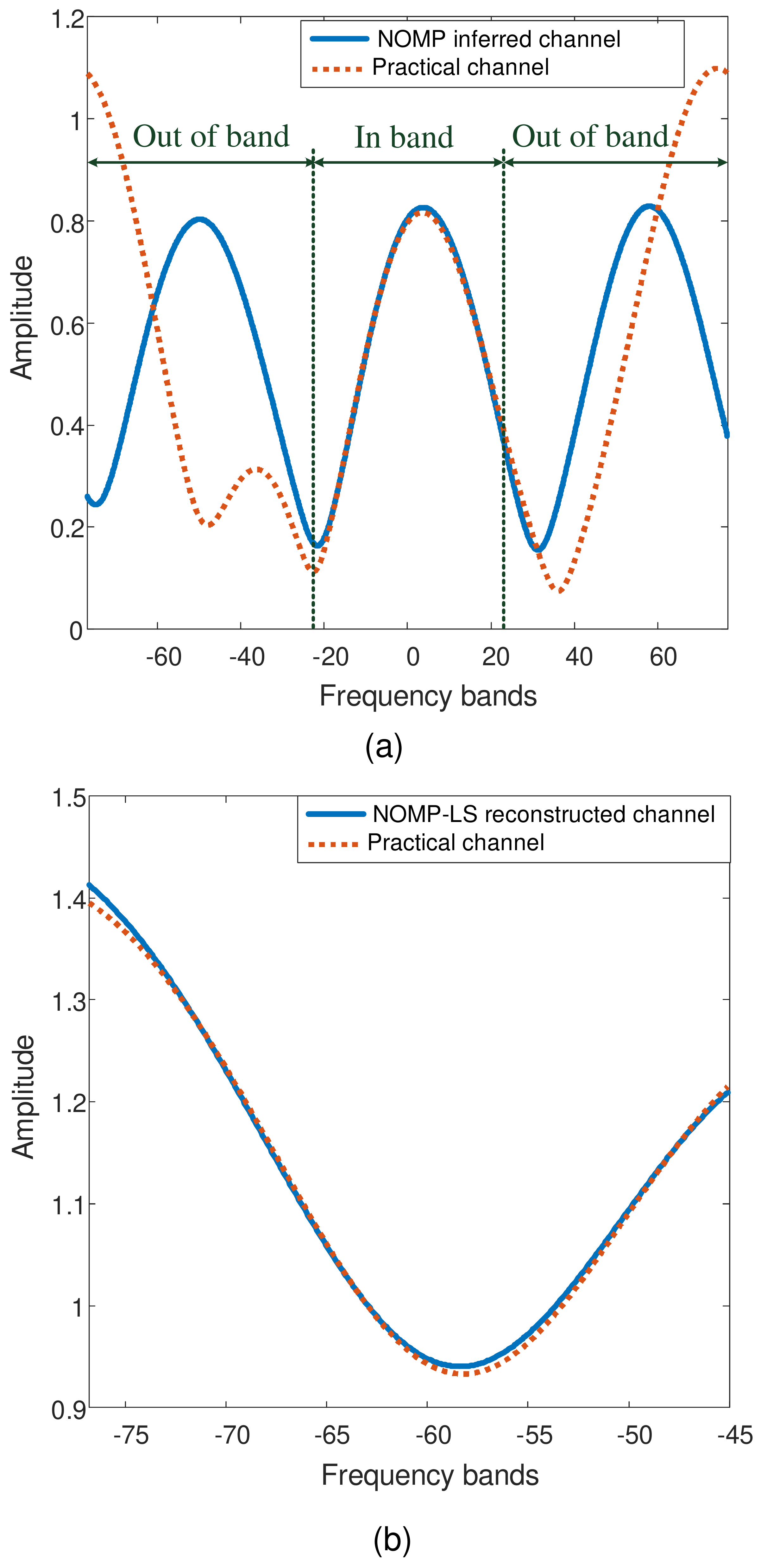}
  \caption{Full-band channel inference and reconstruction results, where (a) presents the inferred channel in the full band, and (b) shows the reconstructed channel in the out-of-band.} \label{Fig:Software FB Results}
\end{figure}

The necessity of the gain refinement is first validated through a comparison of the reconstructed channel's amplitude with that of the real channel. The total bandwidth is 153.6 MHz. The center frequency of 45 MHz is regarded as the in-band to estimate the frequency-independent parameters. The out-of-band channels on the other bands are then inferred or reconstructed using these parameters. We consider a simple example where the BS is equipped with one antenna. Fig.~\ref{Fig:Software FB Results}(a) reveals the difference between the actual full-band channel and the inferred channel. Within the 45 MHz in-band, we note that the inferred channel matches well with the actual channel, corroborating the precision of the NOMP algorithm. On the other hand, for the out-of-band, an obvious deviation can be seen between the inferred and actual channels. The large performance degradation indicates that the gains derived from the in-band estimation are insufficiently accurate in inferring the out-of-band channel. Therefore, the gains are refine using LS estimation with the aid of the out-of-band pilots which are inserted in every four subcarriers. The results of the refinement are given in Fig.~\ref{Fig:Software FB Results}(b). Results demonstrate that the reconstructed out-of-band channel matches closely with the actual channel, thereby validating the necessity and effectiveness of the gain refinement.

\begin{figure}
  \centering
  \includegraphics[scale=0.40]{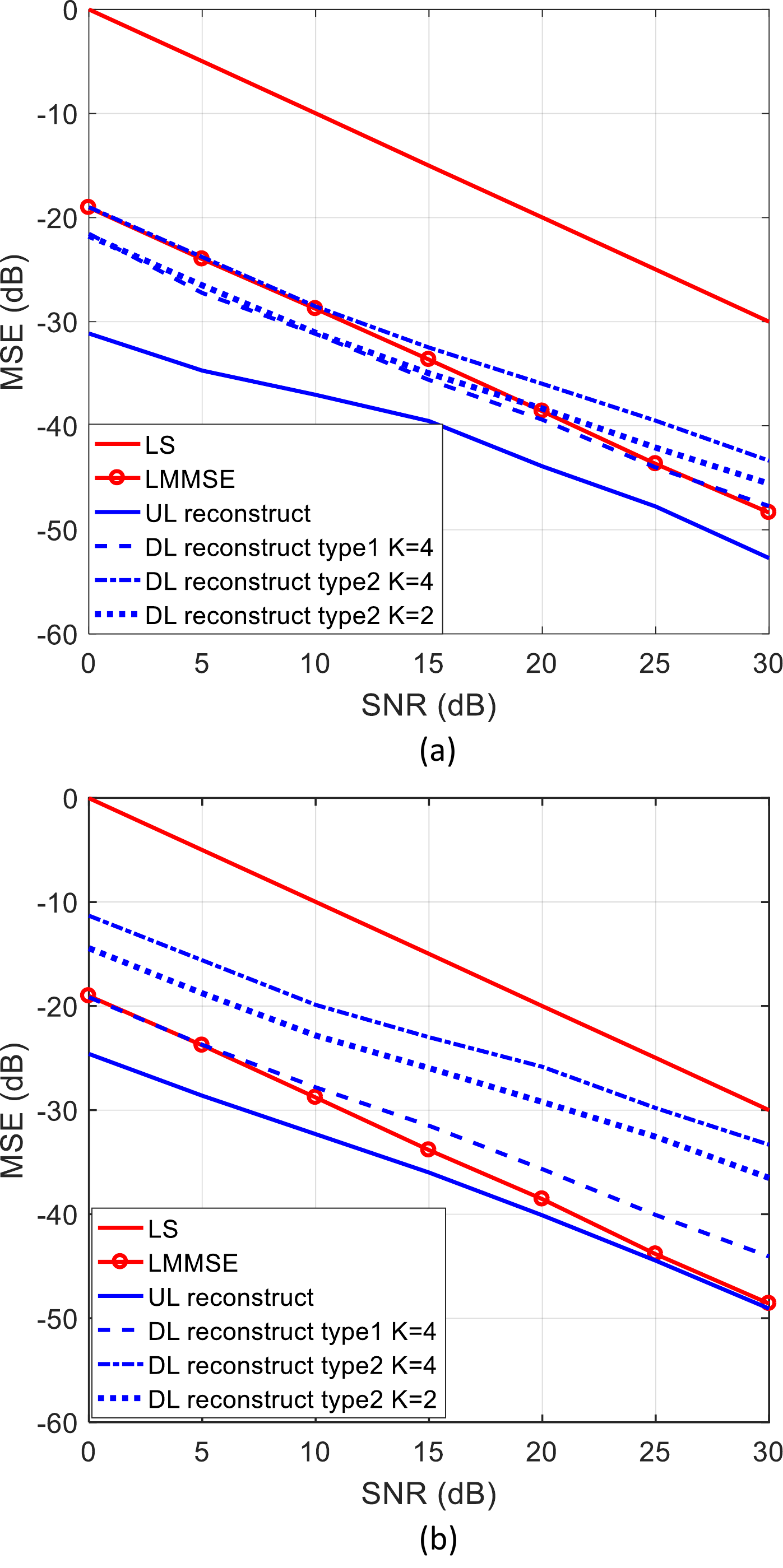}
  \caption{MSE performances of the reconstructions when $M=4$ (a) in Scenario (a), and (b) in Scenario (b).} \label{Fig:SI4O}
\end{figure}

Now, we examine the MSE performance of the proposed efficient downlink channel reconstruction scheme in FDD-OFDM systems. In both uplink and downlink OFDM modules, the central 1200 subcarriers around DC compose the transmission band whose bandwidth equals 90 MHz. The separation between the uplink and downlink central subcarriers is 300 MHz. In the uplink, all 1200 subcarriers are filled with pilots for the NOMP algorithm. Moreover, in the downlink, pilots are sparsely and uniformly inserted in every $K$ subcarriers. We focus on two propagation scenarios. Scenario (a) is a sparsely scattering scenario, where two distinct paths exist in the channel. The angles of the two paths are i.i.d.~and randomly generated in $[0^\circ, 360^\circ)$. Scenario (b) is the clustering channel, where there is one cluster with six close paths. The angular spread of the cluster is $30^\circ$. In addition, the SNR measures the ratio of the pilot power versus the noise power on one antenna and for each subcarrier.

\begin{figure}
  \centering
  \includegraphics[scale=0.40]{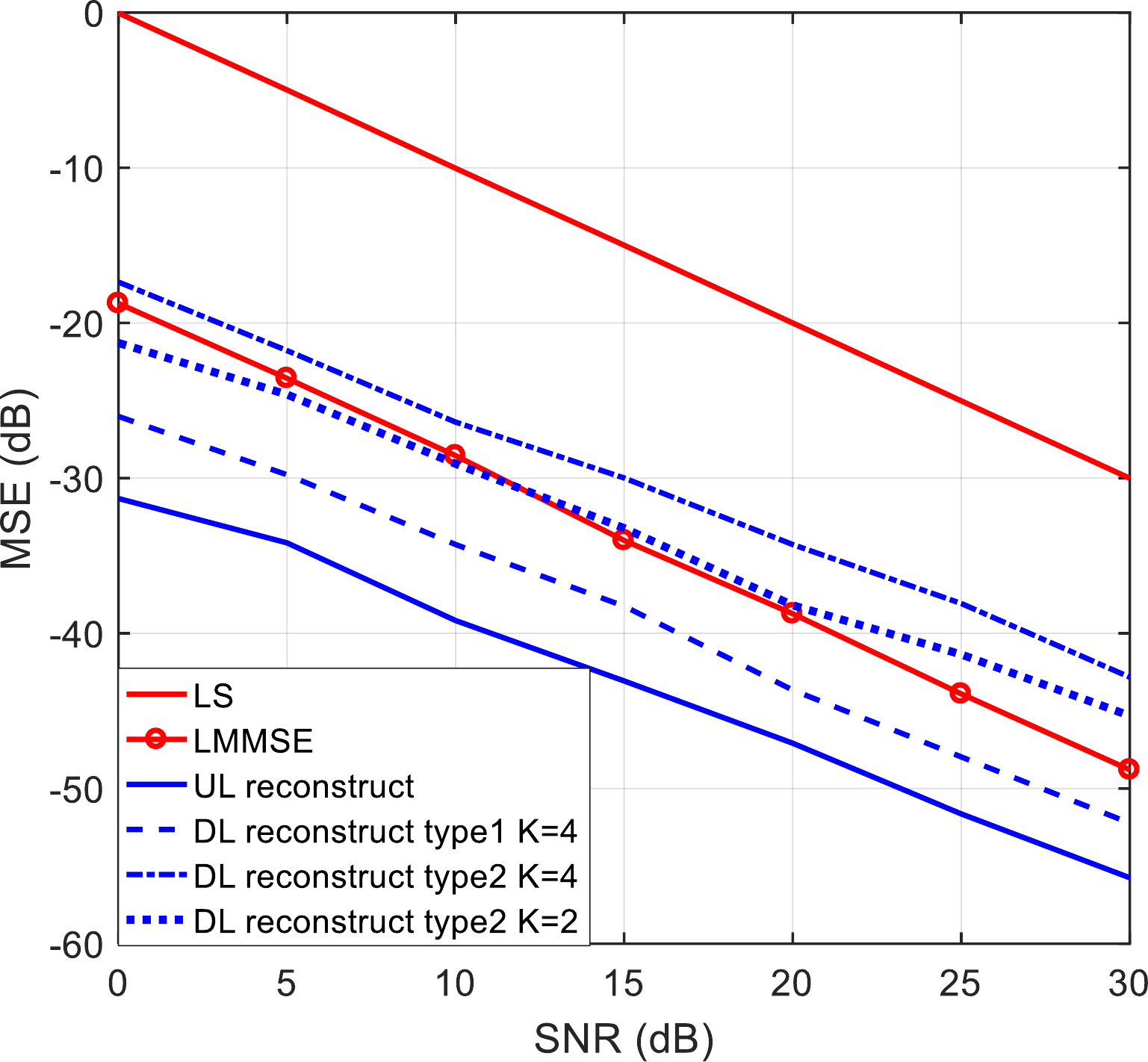}
  \caption{MSE performance of the reconstructions when $M=32$ in Scenario (b).} \label{Fig:SI32Ob}
\end{figure}

The LS and LMMSE channel estimation results are introduced as the lower and upper benchmarks, respectively. LS is a commonly used estimation method with low complexity, but has a drawback of increasing the noise. LMMSE is an improved estimation algorithm that fixes this drawback and achieves considerably higher accuracy. When conducting LS and LMMSE estimation algorithms, we use the pilots on every four subcarriers. We likewise compare the reconstructed downlink and reconstructed uplink channels with the actual channel by evaluating their MSE performance, which is calculated as
\begin{equation}\label{Eq:MSE1}
{\rm{MSE}} =\frac{\mathbb{E}\left\{\|\hat{\bf h}-{\bf h}\|^2\right\}}{\sigma_{\bf z}^2},
\end{equation}
where $\hat{\bf h}$ denotes the reconstructed channel on one subcarrier, ${\bf h}$ is the real channel, and $\sigma_{\bf z}^2$ is the addition of the noise power on multiple antennas which equals $M$ here.

The case when the BS is equipped with four antennas is tested first. Fig.~\ref{Fig:SI4O}(a) demonstrates the MSE performances of the reconstructed uplink and downlink channels in Scenario (a). Results show that the LS estimated channel has the worst MSE, whereas the uplink reconstructed channel has the best MSE. Furthermore, the precision of the uplink-reconstructed channel is even higher than the LMMSE-estimated downlink channel. This finding is attributed to the employment of in-band pilots in the estimation of in-band CSI, and obtaining an accurate composition of the multi-path components through the trivariate NOMP algorithm. Hence, the uplink-reconstructed channel is almost the same as the actual channel. To evaluate the efficient downlink reconstruction scheme, we adopt both beamforming types and compare their MSE performances. As expected, the MSE performance of the downlink reconstruction is inferior to that of the uplink reconstruction. Especially when using beamforming type 2 and setting $K=4$, a significant performance gap appears between the uplink and downlink reconstructions. If we increase the density of downlink pilots by setting $K=2$ or switch to beamforming type 1, the MSE results are improved. As an overly high performance is not necessary and the cost is large, a balance must be reached between performance and cost.

The numerical results of the four-antenna case in Scenario (b) are presented in Fig.~\ref{Fig:SI4O}(b). Extracting each path from their spatial superposition is difficult because the paths are clustered within a small angular-spread area. This condition is particularly true if the angles cannot be accurately estimated, and the number of estimated paths may be more or less than the paths that actual channel has. Hence, the performance of uplink channel reconstruction degrades when compared with that in Fig.~\ref{Fig:SI4O}(a). Regarding the downlink reconstructions, using beamforming type 1 still achieves excellent MSE performance due to the large amount of downlink beamformed pilots. By contrast, using beamforming type 2 results in about 9 dB loss in MSE when compared with type 1 if we set $K=4$ because a relatively large number of estimated paths exist and each estimated direction cannot be allocated with enough pilots. The accuracy is significantly enhanced when setting $K=2$. Therefore, the amount of downlink pilots should be increased in proportion to the number of detected propagation paths.

Now, we scale up the computer simulations by considering the 32-antenna configuration in the more complicated and commonly seen Scenario (b). This simulation aims to assess the performance of the proposed reconstruction scheme in massive MIMO environments. The over-sampling rates are reduced by setting them to 1. The results are shown in Fig.~\ref{Fig:SI32Ob}. Clearly, when the number of antenna elements grows large, the reconstructed channel has excellent performance as well. By comparing Fig.~\ref{Fig:SI32Ob} with Fig.~\ref{Fig:SI4O}(b), we first find that the MSE performance of the uplink reconstruction in the 32-antenna case is obviously better than that in the four-antenna case owing to the high spatial resolution of a large-scale antenna array. With the help of the multi-antenna gain, the downlink reconstructions are significantly improved as well. Therefore, the numerical results indicate that both the uplink reconstruction and the efficient downlink reconstruction perform well in reconstructing the actual channel.

\subsection{OTA Test Results}\label{Sec:OTAResults}

\begin{table}
  \caption{OTA Configurations}
  \label{tab1}
  \centering
  \begin{tabular}{|l|l|}
    \hline
    \bfseries Parameter & \bfseries Value \\[0.5ex]
    \hline
    Antenna Bandwidth & 90 MHz \\
    \hline
    Carrier Frequency & 3.5 GHz \\
    \hline
    Sampling Rate & 153.6 MHz \\
    \hline
    Number of FFT Points & 2048 \\
    \hline
    Subcarrier Spacing & 75 kHz \\
    \hline
    Transmit Power & --20dBm \\
    \hline
  \end{tabular}
\end{table}

\begin{figure*}
  \centering
  \includegraphics[scale=0.44]{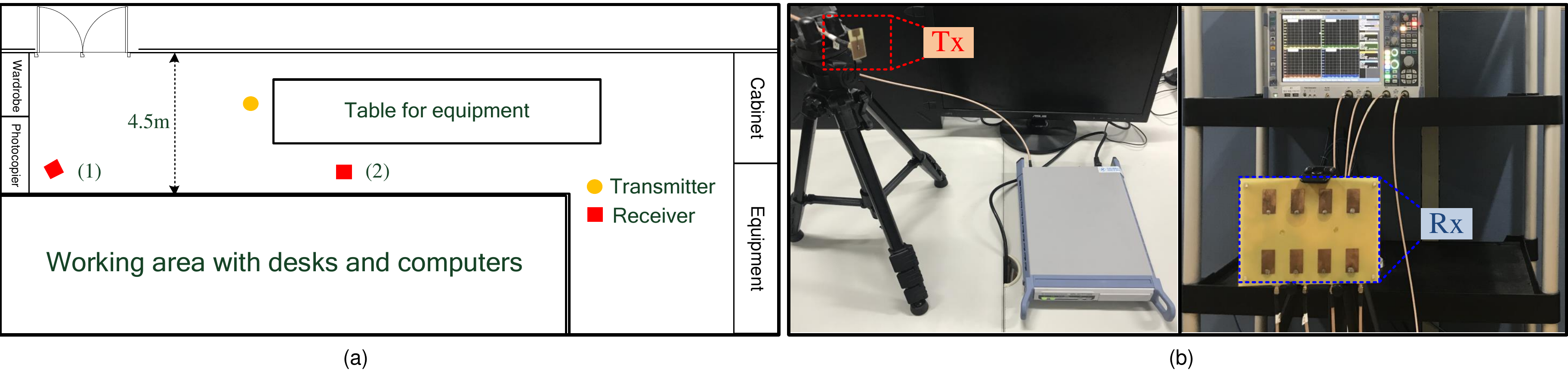}
  \caption{(a) OTA test environment in a laboratory and radio devices placed along the table for equipment. (b) The left and the right subfigures show the user and the BS, respectively.}\label{Fig:Lab}
\end{figure*}


We also set up an OTA testbed [Fig.~\ref{Fig:Lab}(a)] to validate the results in practical environments. The radio devices are placed along the table for equipment. The yellow circle represents the position of the user, and the red squares are the potential positions of the BS. The BS and the user are equipped as shown in Fig.~\ref{Fig:Lab}(b). The user works as the transmitter and has a single antenna controlled by a RF vector signal generator. The BS is the receiver. The received signal at the BS antenna array is first transported to a digital oscilloscope. After down-converting, synchronizing, and sampling, the received signal is imported to the computer and processed through MATLAB. Fig.~\ref{Fig:Lab}(b) illustrates the BS antenna array, which is a four-element ULA where one column of the array is combined to form a ULA element. When evaluating the single-antenna case, the ULA is replaced by one antenna element like the user antenna.

The configurations of the OTA tests are listed in Table \ref{tab1}. Owing to the limitations of the hardware equipments, in-band versus out-of-band tests are used to imitate the uplink versus downlink tests. Considering the antenna bandwidth , we select the in-band and out-of-band regions within the central 90 MHz band and separate them to the greatest extent. As shown in Fig.~\ref{Fig:HardwareBands}, we regard the red region with 45 MHz bandwidth as the in-band to imitate the uplink. The 15 MHz-bandwidth region colored in blue is chosen as the out-of-band to imitate the downlink. The central frequency of the out-of-band region is 60 MHz away from the central frequency of the in-band region. In the gain refinement stage, $K=2$ or 4, which means that one-half or one-fourth of the subcarriers are allocated for the out-of-band pilots. Given the high accuracy of LMMSE estimation algorithm, we regard the LMMSE-estimated channel as the real channel when evaluating the out-of-band reconstruction scheme. The out-of-band channel inference method is also evaluated, which represents the method introduced in \cite{Vasisht2016Eliminating}. Similarly, MSE is used as the metric, which is calculated as
\begin{equation}\label{Eq:MSE2}
{\rm MSE}=\frac{\mathbb{E}\left\{\left\|{\hat {\bf h}}-{\hat {\bf h}}_{\rm LMMSE}\right\|^2\right\}}{\sigma_{\bf z}^2},
\end{equation}
where ${\hat {\bf h}}_{\rm LMMSE}$ is the LMMSE-estimated channel and is regarded as the real channel. The test results are displayed in the form of a cumulative distribution function (CDF).

We start by reconstructing the simplest channel when the BS is equipped with a single antenna. In the single-antenna tests, $K=4$. Fig.~\ref{Fig:HardwareSISO1}(a) provides the CDF of the MSE (in dB) when the BS is located at Position (1). The figure shows that in-band reconstruction achieves the highest accuracy, with a 90\% probability that the MSE is below $-9.65$ dB. However, the performance of the out-of-band inference is poor, with a 90\% probability that MSE is lower than $15.36$ dB, demonstrating that the inferred channel can not accurately depict the actual channel. Fortunately, the accuracy is greatly improved when the gains are refined with the aid of the out-of-band pilots. The out-of-band reconstruction scheme functions well, with a 90\% probability that the MSE is below $-6.64$ dB. These OTA results align with the previous numerical results. We can further investigate the power ratio of the propagation paths through the results in Fig.~\ref{Fig:HardwareSISO1}(b). The first detected path occupies 75.2\% power of the channel because a strong line-of-sight (LoS) propagation path can be detected at Position (1). The power ratio increases to 94.3\% after the second detected paths is added. The results also indicate that the possibility of detecting more than four paths is below 0.036.

\begin{figure}
  \centering
  \includegraphics[scale=0.5]{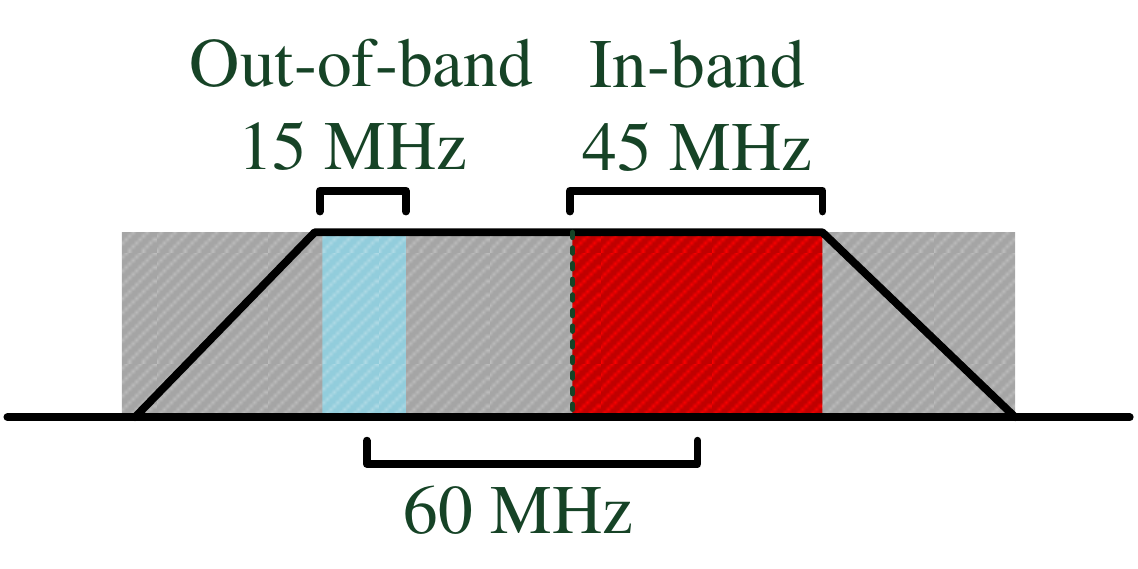}
  \caption{In-band (red region) and out-of-band (blue region) in OTA tests.}\label{Fig:HardwareBands}
\end{figure}

\begin{figure*}
  \centering
  \includegraphics[width=1.00\linewidth]{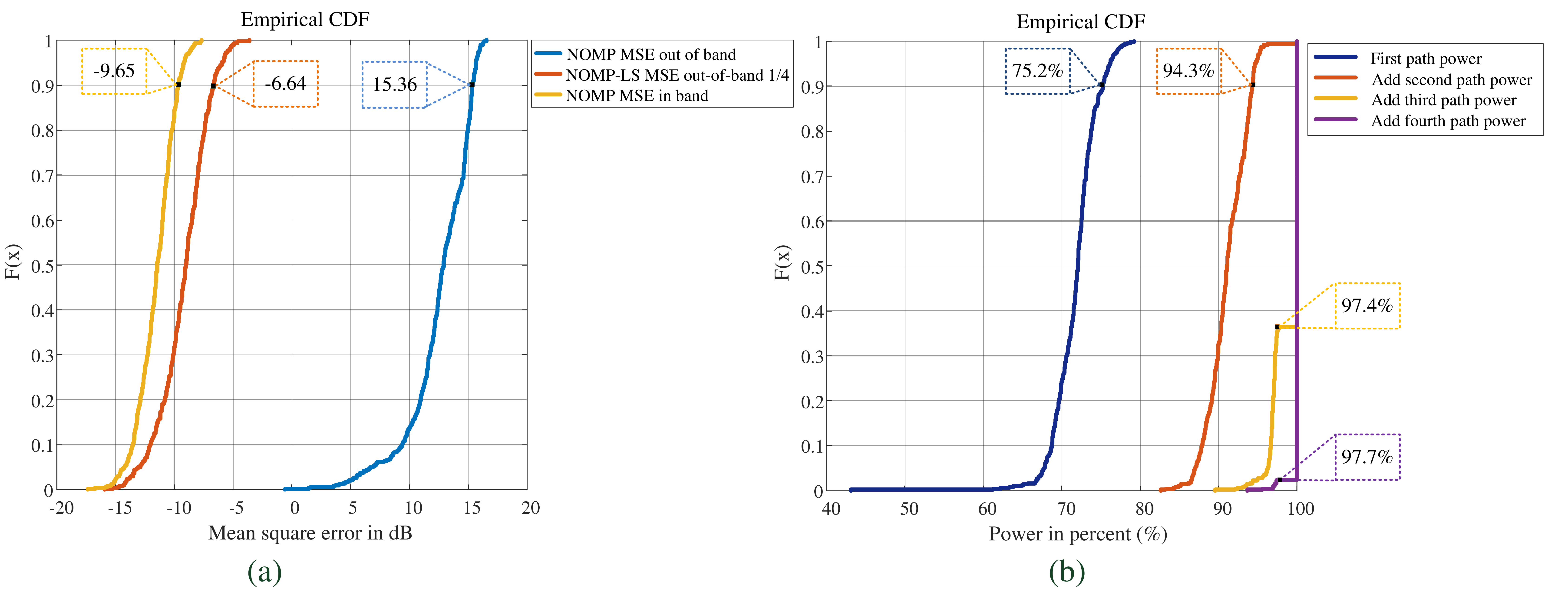}
  \caption{SISO: MSE performances of the channel reconstructions and the power ratio of reconstructed paths when conducting OTA tests in Position (1).}\label{Fig:HardwareSISO1}
\end{figure*}

\begin{figure*}
  \centering
  \includegraphics[width=1.00\linewidth]{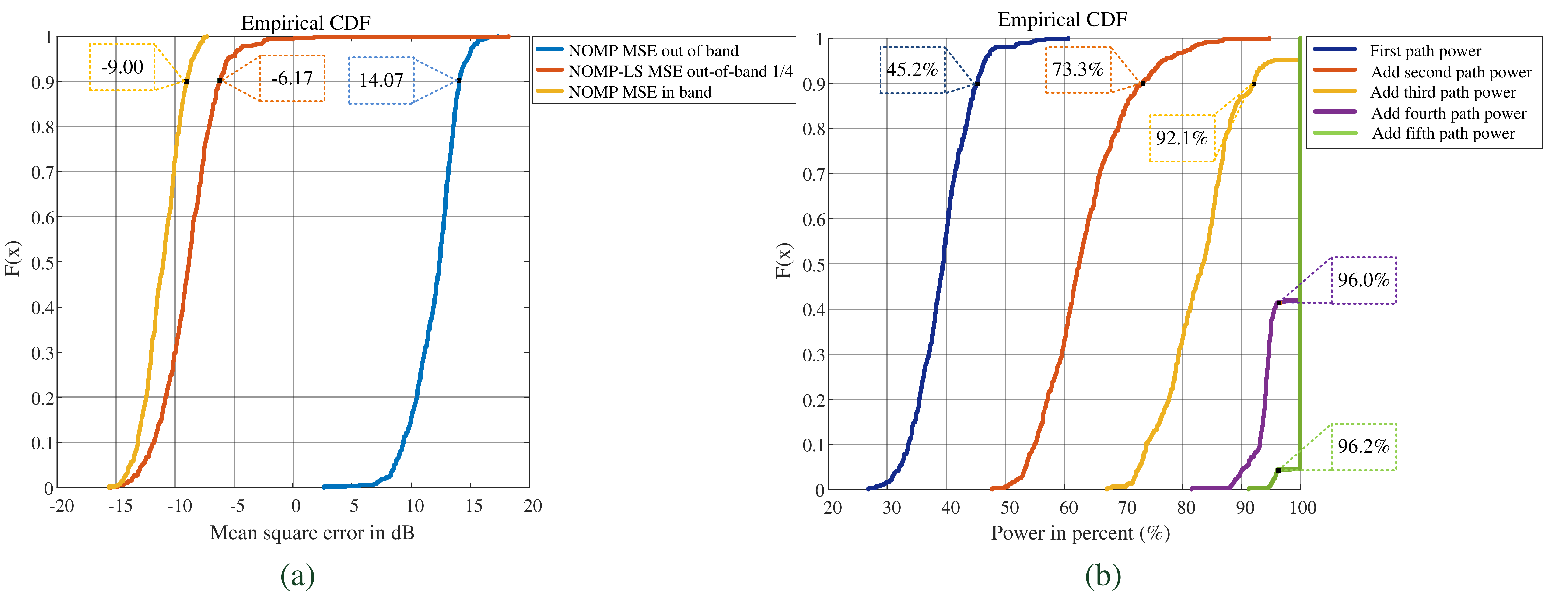}
  \caption{SISO: MSE performances of the channel reconstructions and the power ratio of reconstructed paths when conducting OTA tests in Position (2).}\label{Fig:HardwareSISO2}
\end{figure*}

The performances are then tested when the BS is located at Position (2), and the results are given in Fig.~\ref{Fig:HardwareSISO2}. We find that both in-band and out-of-band reconstructions still have excellent performances in the non-LoS (NLoS) propagation scenario. The 90\%-probability MSEs of the two reconstructions are $-9.00$ and $-6.17$ dB, respectively. Although the MSE performance is inferior to that of the LoS case, this performance degradation is relatively small. As for the power ratio of the propagation paths, the first reconstructed path only occupies 45.2\% power of the channel and has a 0.416 probability that the number of existing paths is more than four. These results clearly reveal that the out-of-band reconstruction scheme functions well even when multiple NLoS propagation paths exist in the channel.

\begin{figure*}
	\centering
	\includegraphics[width=1.00\linewidth]{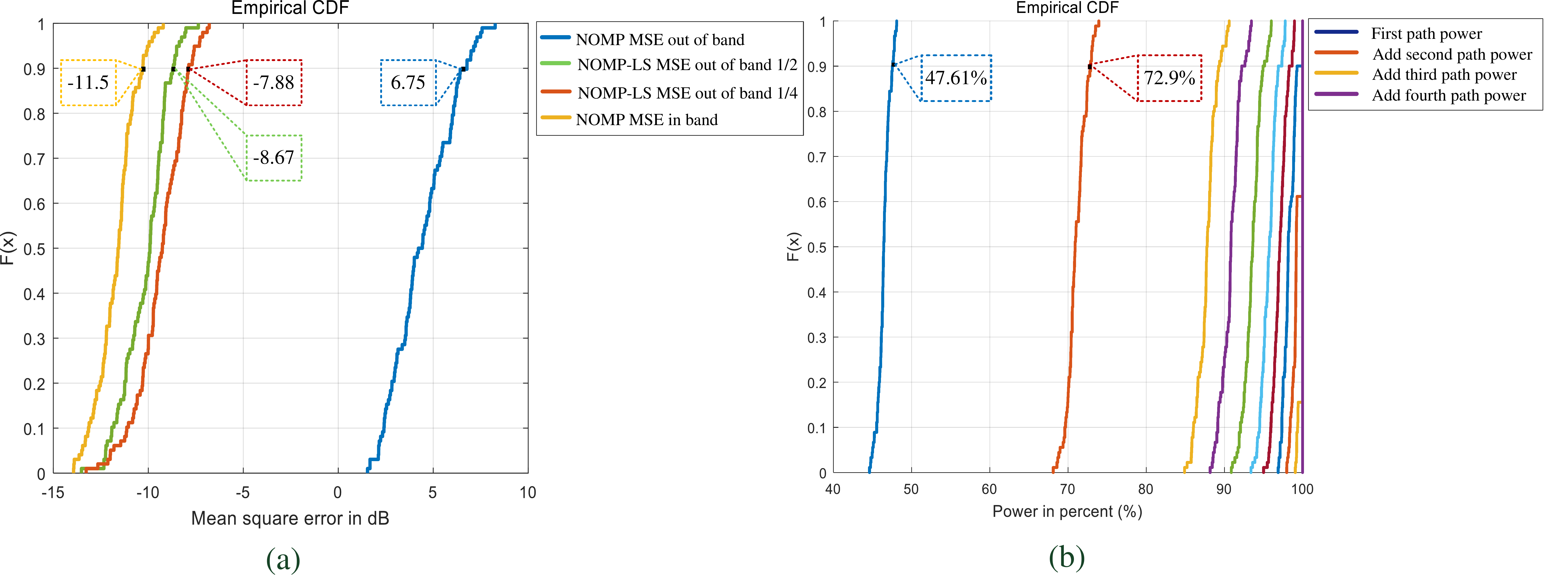}
	\caption{SIMO: MSE performances of the reconstructions and the power ratio of reconstructed paths when conducting OTA tests in Position (1).}\label{Fig:HardwareSI4O1}
\end{figure*}

\begin{figure*}
	\centering
	\includegraphics[width=1.00\linewidth]{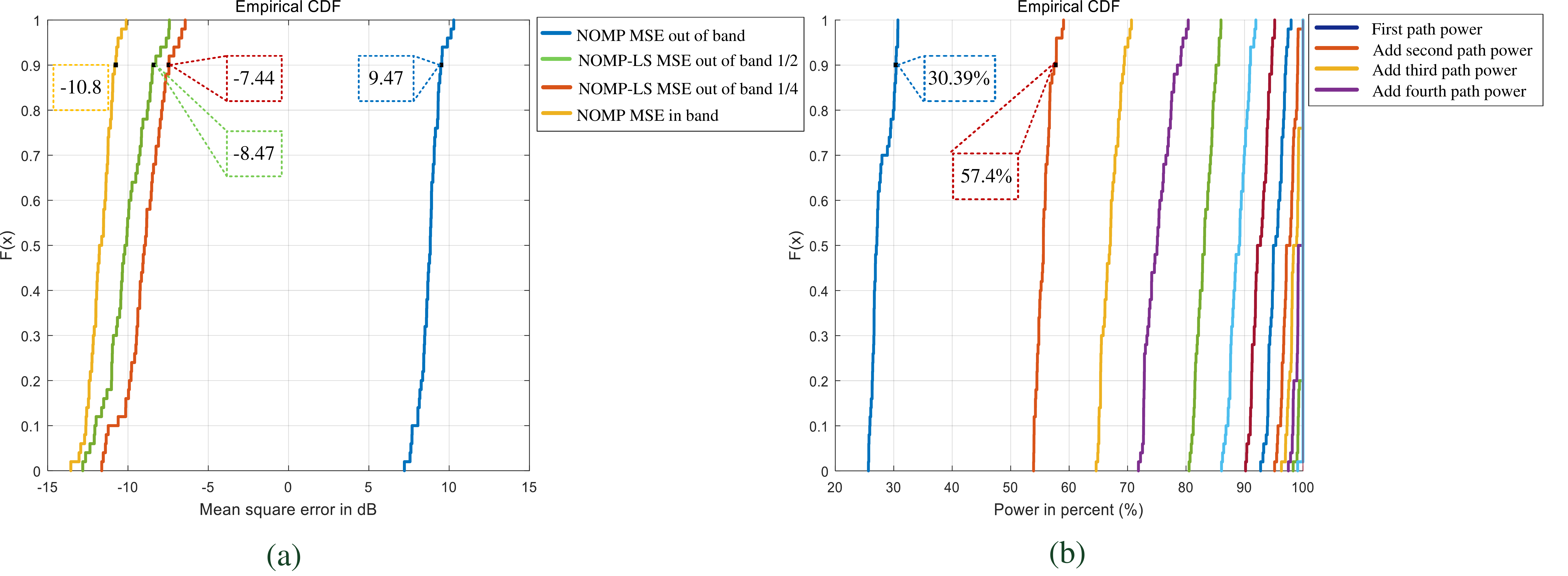}
	\caption{SIMO: MSE performances of the reconstructions and the power ratio of reconstructed paths when conducting OTA tests in Position (2).}\label{Fig:HardwareSI4O2}
\end{figure*}

Next, the performance of the out-of-band reconstruction is evaluated in the single-input multiple-output (SIMO) system where the BS is configured to have four antennas.
The CDF of the MSE results when the BS is located at Position (1) is shown in Fig.~\ref{Fig:HardwareSI4O1}. Different from the single-antenna cases, in the multi-antenna case, the BS needs to estimate the angles, in addition to the delays and gains. The channel tested from Position (1) is similar to Scenario (b) in computer simulations. A dominant LoS propagation path exists, but it is surrounded by multiple paths. These paths compose a cluster, which makes it more difficult to separate these paths from one another. Fig.~\ref{Fig:HardwareSI4O1}(b) shows that the number of estimated paths increases greatly when compared with that in Fig.~\ref{Fig:HardwareSISO1}(b), even though these results are derived at the same place. This finding is attributed to the multiple antennas that enhance the spatial resolution and explain the spatial channel in a more detailed way. Therefore, the in-band channel reconstruction achieves higher accuracy in a multiple-antenna environment than in a single-antenna environment. When the out-of-band channel reconstruction scheme is used, pilots are inserted in every two or four subcarriers at the gain refinement stage. The reconstruction accuracy for the out-of-band is improved when the density of the pilots is increased. It has a 90\% probability that the MSE is below $-8.67$ dB if we set $K=2$. These results strongly reveal the spatial reciprocity between the channels in separated bands and validate the effectiveness of the proposed reconstruction scheme.

Fig.~\ref{Fig:HardwareSI4O2} presents the results in the setup where the BS is located at Position (2) and the channel has NLoS paths. A slight but negligible performance degradation can be found from the MSE lines of both in-band reconstruction and out-of-band reconstruction when these are compared with Fig.~\ref{Fig:HardwareSI4O1}. Additionally, the power ratio of the first two detected paths decreases. These observations are in accordance with those from Figs.~\ref{Fig:HardwareSISO1} and \ref{Fig:HardwareSISO2}, thus demonstrating the correctness of the OTA tests regarding the proposed reconstruction scheme and indicating that the reconstruction scheme can work well in cases with more antennas.

\section{Conclusion}\label{Sec:Conclusion}

In this study, an efficient downlink channel reconstruction scheme was proposed for FDD multi-antenna systems. The scheme uses the frequency-independent features of the spatial parameters and tackles the problem of downlink CSI acquisition at the BS in the absence of uplink-downlink reciprocity with limited overhead of downlink training and feedback. We extended the NOMP algorithm to cope with the multi-antenna multi-subcarrier condition for extracting the frequency-independent parameters. Numerical simulations validated the effectiveness of the gain refinement, which causes the downlink training and feedback overhead. Our OTA tests demonstrated that the downlink reconstruction scheme achieves promising MSE performance. The scheme and the OTA results have directive significance to the design of FDD massive MIMO systems.

\end{document}